\documentclass{sig-alternate-05-2015}
\usepackage{times}
\usepackage{amsmath,amssymb,amsfonts,bm}
\usepackage[ruled]{algorithm2e}
\usepackage{graphicx}
\usepackage[tight,normal]{subfigure}
\usepackage{multirow,url,footmisc}
\newtheorem{definition}{Definition}
\newtheorem{theorem}{Theorem}

\begin{document}

\CopyrightYear{2016}
\setcopyright{acmcopyright}
\conferenceinfo{SIGIR '16,}{July 17--21, 2016, Pisa, Italy}
\isbn{978-1-4503-4069-4/16/07}
\acmPrice{\$15.00}
\doi{http://dx.doi.org/10.1145/2911451.2911493}

\clubpenalty=10000 
\widowpenalty = 10000

\title{Composite Correlation Quantization for Efficient Multimodal Retrieval\titlenote{Corresponding authors: Jianmin Wang and Mingsheng Long.}}

\numberofauthors{1}
\author{
\alignauthor
Mingsheng Long$^{\dag}$, Yue Cao$^{\dag}$, Jianmin Wang$^{\dag}$, and Philip S. Yu$^{\ddag\sharp}$\\
\affaddr{$^\dag$School of Software, Tsinghua National Laboratory (TNList), Tsinghua University, Beijing, China}\\
\affaddr{$^\ddag$Institute for Data Science, Tsinghua University \quad $^\sharp$University of Illinois at Chicago, IL, USA}\\
\email{\{mingsheng, jimwang\}@tsinghua.edu.cn, caoyue10@gmail.com, psyu@uic.edu}\\
}

\maketitle

\begin{abstract}
Efficient similarity retrieval from large-scale multimodal database is pervasive in modern search engines and social networks. To support queries across content modalities, the system should enable cross-modal correlation and computation-efficient indexing. While hashing methods have shown great potential in achieving this goal, current attempts generally fail to learn isomorphic hash codes in a seamless scheme, that is, they embed multiple modalities in a continuous isomorphic space and separately threshold embeddings into binary codes, which incurs substantial loss of retrieval accuracy. In this paper, we approach seamless multimodal hashing by proposing a novel Composite Correlation Quantization (CCQ) model. Specifically, CCQ jointly finds correlation-maximal mappings that transform different modalities into isomorphic latent space, and learns composite quantizers that convert the isomorphic latent features into compact binary codes. An optimization framework is devised to preserve both intra-modal similarity and inter-modal correlation through minimizing both reconstruction and quantization errors, which can be trained from both paired and partially paired data in linear time. A comprehensive set of experiments clearly show the superior effectiveness and efficiency of CCQ against the state of the art hashing methods for both unimodal and cross-modal retrieval.
\end{abstract}


\section{Introduction}\label{section:Introduction}
While big data with large volume, high dimensions, and multiple modalities are ubiquitous in search engines and social networks, it has attracted increasing attention to distill the correlation structures across heterogenous data modalities. For example, an uploaded image on Flickr is usually annotated with some relevant descriptions or tags, while a featured article on Wikipedia may consist of some correlative images. As relevant data from different modalities may endow semantic correlations, it is desirable to support \emph{multimodal search}, which retrieves semantically-relevant results of all modals in response to a unimodal query. Taking Flickr as an example, when a query image is given, the system should return both relevant tags and images. Due to large volume and semantic gap \cite{cite:TPAMI14Wiki}, effective and efficient retrieval of multimodal data remains a challenge.

In the case that the reference database is large-scale or that the distance calculation between query item and database item is costly, an efficient solution to enabling similarity search is hashing based methods \cite{cite:Arxiv14HashSurvey}, which perform approximate nearest neighbor (ANN) search with both computation efficiency and acceptable accuracy. The principle of hashing is to transform high-dimensional data into compact binary codes and generate similar binary codes for similar data items. The seminal work includes Locality Sensitive Hashing (LSH) \cite{cite:FOCS06LSH} and Spectral Hashing (SH) \cite{cite:NIPS09SH}. However, traditional \emph{unimodal} hashing methods cannot support multimodal search as ANN cannot be directly computed across different modalities.

Recently, several useful attempts have been made to \emph{multimodal hashing}, which builds correlation structures across multiple modalities in the process of hash function learning and index multimodal data in a common Hamming space \cite{cite:CVPR10CMSSH,cite:SIGIR11CHMIS,cite:IJCAI11CVH,cite:NIPS12CRH,cite:KDD12MLBE,cite:MM13LCMH,cite:SIGMOD13IMH,cite:VLDB14MSAE,cite:SIGIR14DCDH,cite:AAAI14SCM,cite:MM14CorrAE,cite:IJCAI15QCH,cite:ICCV15MCNN}. These methods generally work in two-step pipeline: first, embed multiple data modalities into a \emph{continuous} isomorphic latent space by maximizing inter-modal correlations, and second, quantize the isomorphic embeddings into binary hash codes by sign thresholding. While showing promising performance, the two-step pipeline may encounter two limitations: first, conversion from real-valued features to discrete codes may incur substantial information loss, making the continuous latent space suboptimal for binary coding and the binary codes suboptimal for retrieval \cite{cite:VLDB14MSAE,cite:MM14IMVH}; second, directly binarizing latent features may lead to unbalanced encoding schemes \cite{cite:NIPS12CRH,cite:KDD12MLBE}. Fundamentally, by continuous relaxation of the binary constraints, most methods solve an optimization problem which may deviate significantly from the hashing objective as the quantization error is not accounted for in the optimization process. This somewhat contradicts the motivation of multimodal hashing. Hence, how to learn isomorphic hash codes for multimodal data in a seamless optimization framework remains an open problem.

\begin{figure*}[tbp]
    \centering
    \includegraphics[width=1.0\textwidth]{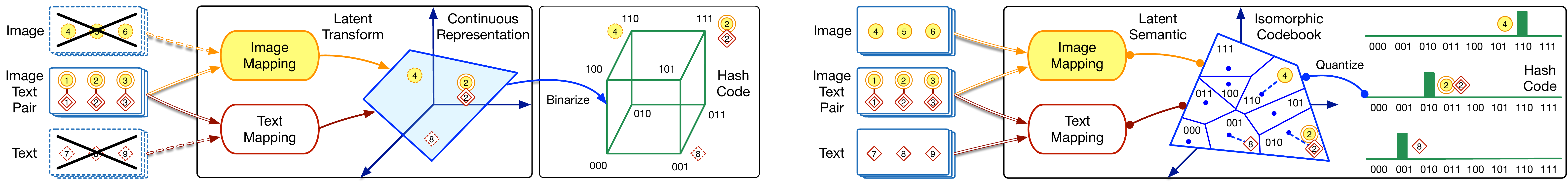}
    \caption{Flowcharts of prior work (left) and CCQ (right). Prior work is a two-step pipeline: first map image-text pairs to isomorphic latent space (denoted as polygon) and then binarize the continuous representation to hash codes (denoted as vertices of hypercube) by \emph{sign thresholding}. CCQ is a seamless optimization framework: jointly map both paired/unpaired images and texts to isomorphic latent space (denoted as polygon) and learn hash codes by \emph{composite quantization}. The quantization model learns isomorphic codebook (denoted as Voronoi digram) and binary codes (denoted as histograms) by minimizing the quantization error, which suffices to assign each latent representation to $M$-nearest codewords (denoted as Voronoi cells) and assignment indices are used as hash codes.}
    \label{fig:CCQi}
    \vspace{-5pt}
\end{figure*}

In this paper, we propose Composite Correlation Quantization (CCQ), a novel model towards seamless multimodal hashing. Technically, CCQ jointly finds correlation-maximal mappings that transform different modalities into an isomorphic latent space, and learns composite quantizers that convert the isomorphic latent features into compact binary codes. The flowcharts of CCQ and prior work are shown in Figure~\ref{fig:CCQi}. To create a seamless optimization framework, we are inspired by Latent Semantic Analysis (LSA) \cite{cite:IS90LSA} and decompose each datum into three latent factors, namely, correlation-maximal mapping, similarity-preserving codebook, and compact binary code. The three latent factors are jointly learned through an optimization problem, which preserves both intra-modal similarity and inter-modal correlation while minimizing both reconstruction and quantization errors. The CCQ model can construct extremely compressed and balanced binary codes to enable efficient multimodal search, can readily handle a ubiquitous semi-paired scenario where only a fraction of input data are multimodal, and can scale linearly to large sample size. Comprehensive empirical evidence on large-scale datasets confirms that the CCQ model exhibits superior performance in both effectiveness and efficiency on both unimodal and cross-modal search against state of the art hashing methods.

The subsequent paper is organized as follows. We review related works in Section~\ref{section:RelatedWork}. We formally present our model in Section~\ref{section:CCQ} and algorithm with analysis in Section~\ref{section:Optimization}. Empirical evaluations are reported in Section~\ref{section:Experiments}, while conclusions are enclosed in Section~\ref{section:Conclusion}.

\section{Related Work}\label{section:RelatedWork}
Recently, hashing-based multimodal search is a prevalent research focus in machine learning and information retrieval communities \cite{cite:CVPR10CMSSH,cite:IJCAI11CVH,cite:MM13LCMH,cite:SIGMOD13IMH,cite:SIGIR14DCDH,cite:MM14CorrAE,cite:MM14IMVH,cite:AAAI14SCM,cite:IJCAI15QCH,cite:IJCAI15PMMH,cite:CVPRDSRH}, which enables approximate similarity search on multimedia database with significant speedup and acceptable accuracy. Refer to \cite{cite:Arxiv14HashSurvey} for a comprehensive survey.

Existing multimodal hashing methods can be organized into two categories: supervised methods and unsupervised methods. CMSSH \cite{cite:CVPR10CMSSH}, SCM \cite{cite:AAAI14SCM}, QCH \cite{cite:IJCAI15QCH}, and SePH \cite{cite:CVPR15SPH} are supervised hashing methods that require labeled pairs to indicate if the objects from different modalities are similar (positive) or dissimilar (negative). As supervised information is usually unavailable in many applications, the deployment of these methods may be severely restricted. CVH \cite{cite:IJCAI11CVH}, IMH \cite{cite:SIGMOD13IMH}, MSAE \cite{cite:VLDB14MSAE} and CorrAE \cite{cite:MM14CorrAE} are unsupervised hashing methods applicable to the most general multimodal retrieval case given that paired data are available, while our proposed CCQ model falls into this category. IMH \cite{cite:SIGMOD13IMH} is an extension of spectral hashing \cite{cite:NIPS09SH} to multimodal data, which is restricted by the training burden since constructing and eigendecomposing the similarity matrices require $O(N^2)$. While CVH \cite{cite:IJCAI11CVH} tackles the scalability issue, it does not jointly maximize cross-modality correlation and preserve intra-modality similarity. MSAE \cite{cite:VLDB14MSAE} and CorrAE \cite{cite:MM14CorrAE} can capture both intra-modal similarity and inter-modal correlation by deep autoencoders, but they require spectral hashing or sign thresholding for obtaining binary codes from the continuous embeddings, which will give rise to uncontrollable quantization errors \cite{cite:CVPR11ITQ,cite:TPAMI11PQ}.

A crucial problem with existing methods is that they essentially work in a separated two-step pipeline: first embed multimodal data into a common \emph{continuous} latent space and then threshold the continuous embeddings into binary codes of the Hamming space. Such conversion from real-valued features to discrete codes may result in substantial information loss, making the continuous latent space suboptimal for the binary codes and the binary codes suboptimal for retrieval \cite{cite:ICML14CQ}. Furthermore, directly binarizing latent representation may lead to unbalanced encoding schemes, as shown in \cite{cite:NIPS12CRH,cite:KDD12MLBE}. Although IMVH \cite{cite:MM14IMVH} learns multimodal hash functions using a graph-cut quantizer instead of the sign thresholding, the quantizer solves a fast approximation of energy function with orthogonal constraints and recurs large quantization error and unbalanced codes. CCQ approaches this problem by learning the modality-consistent latent space and balanced binary codes in a principled framework.

\section{Composite Correlation Quantization}\label{section:CCQ}

\subsection{Problem Statements}
In the multimodal search system, the database and query consist of objects from different modalities. We only use image and text as two modalities to explain our approach, but the approach is formulated to support any number $V$ of modalities. Let ${\bf X}^1 \in \mathbb{R}^{P_1 \times N_1}$ be an image set of $N_0$ images with tags and the rest ${\bar N_1}$ images without tags, where $N_1 = N_0 + {\bar N}_1$ and each image is represented by $P_1$-dimensional feature vector. Let ${\bf X}^2 \in \mathbb{R}^{P_2 \times N_2}$ be a text set of $N_0$ documents of the image tags and additional ${\bar N}_2$ documents, where $N_2 = N_0 + {\bar N}_2$ and each text is represented by $P_2$-dimensional feature vector. Note that the proposed approach can handle \emph{semi-paired} data where only a fraction $N_0/({N}_1 + {N}_2)$ of objects are multimodal, and is more realistic than typical multimodal methods.

An efficient approach to calculating the distance between image and text is to map images and texts to modality-isomorphic binary codes in which different modalities of the objects are comparable. In this paper, we will approach this problem by a joint optimization framework, dubbed Composite Correlation Quantization (CCQ).

\begin{definition}[CCQ]
	Given an image ${\bf x}^1_n \in \mathbb{R}^{P_1}$ and a text ${\bf x}^2_n \in \mathbb{R}^{P_2}$, learn two correlation-maximal mappings ${f}^1 : \mathbb{R}^{P_1} \mapsto \mathbb{R}^{D}$ and ${f}^2 : \mathbb{R}^{P_2} \mapsto \mathbb{R}^{D}$ that transform images and texts into a $D$-dimensional isomorphic latent space, and jointly learn two composite quantizers ${q}^1 : \mathbb{R}^{D} \mapsto \{0, 1\}^{H}$ and ${q}^2 : \mathbb{R}^{D} \mapsto \{0, 1\}^{H}$ that quantize latent embeddings into compact $H$-bits binary codes.
\end{definition}

In the common $H$-bits binary space, image and text can be easily comparable such that both intra-modal and cross-modal search can be readily supported. After mappings $f^1$, $f^2$ and quantizers $q^1$, $q^2$ have been learned, the multimodal search problem can be converted into classical approximate nearest neighbor (ANN) search problem.

\subsection{Composite Correlation Quantization}
The main idea of CCQ is to jointly learn a correlation-maximal latent space and a similarity-preserving composite quantization in a unified optimization framework. To achieve this mission, we are inspired by Latent Semantic Analysis (LSA) \cite{cite:IS90LSA} and decompose each input datum (image or text) ${\bf x}^v_n$ into three latent factors ${\bf R}^v$, ${\bf C}^v$, ${\bf b}^v_n$, that is, ${\bf{x}}_n^v \approx {{\bf{R}}^v}{{\bf{C}}^v}{\bf{b}}_n^v$. While sharing similar formation as LSA, our formulation endows these latent factors with different semantics and thus constrains them with different conditions. More specifically, ${\bf R}^v$ is correlation-maximal mapping, ${\bf C}^v$ is similarity-preserving codebook, and ${\bf b}^v_n$ is the compact binary code of ${\bf x}^v_n$. We present how to formulate the CCQ approach under these semantics.

\subsubsection{Intra-Modality Similarity Quantization}
To represent inputs with compact binary codes, two mainstream paradigms are sign thresholding in Hamming embedding methods \cite{cite:NIPS09SH}, and vector quantization in codebook-based encoding methods \cite{cite:TPAMI11PQ}. As sign thresholding cannot guarantee minimal quantization error, we therefore adopt the vector quantization paradigm. CCQ is based on a set of $M$ codebooks ${\bf C}^v = [{\bf C}^v_1,\ldots,{\bf C}^v_M]$, where each codebook ${\bf C}^v_m$ contains $K$ codewords ${\bf C}^v_m = [{\bf C}^v_{m1},\ldots,{\bf C}^v_{mK}]$, and each codeword ${\bf C}^v_{mk}$ is a $D$-dimensional vector like the cluster centroid in kmeans clustering. Corresponding to the $M$ codebooks, we partition the binary codewords assignment vector ${\bf b}^v_n$ into $M$ $1$-of-$K$ indicator vectors ${\bf b}^v_n = [{\bf b}^v_{1n}; \ldots; {\bf b}^v_{mn}]$, and each indicator vector ${\bf b}^v_{mn}$ indicates which one (and only one) of the $K$ codewords in the $m$th codebook is selected to approximate the $n$th data point. The CCQ model encodes each ${\bf x}^v_n$ as the sum of $M$ codewords, one codeword per codebook, each indicated by the binary assignment vector ${\bf b}^v_n$. This yields a novel and more accurate composite approximation scheme ${\bf{x}}_n^v \approx {{\bf{R}}^v}\sum\nolimits_{m = 1}^M {{\bf{C}}_m^v{\bf{b}}_{mn}^v} $. Consistent with LSA and kmeans, the sum of squared loss between all ${\bf x}^v_n$'s and the sum of selected codewords after transformed by ${\bf R}^v$, is minimized,
\begin{equation}\label{eqn:CCQintra}
  \begin{aligned}
		\mathop {\min }\limits_{{{\bf{R}}^v},{{\bf{C}}^v},{{\bf{B}}^v}} & \sum\limits_{n = 1}^{{N_v}} {\left\| {{\bf{x}}_n^v - {{\bf{R}}^v}\sum\limits_{m = 1}^M {{\bf{C}}_m^v{\bf{b}}_{mn}^v} } \right\|_2^2} \\
    {\rm{s.t.}} \quad &{\left\| {{{\bf{b}}_{{m}n}}} \right\|_0} = 1,{{\bf{b}}_{{m}n}} \in {\left\{ {0,1} \right\}^K} \\
    &m = 1 \ldots M,n = 1 \ldots N_v, \\ 
  \end{aligned}
\end{equation}
where ${\left\|  \cdot  \right\|_0}$ denotes the $\ell_0$-norm that simply counts the number of the vector's nonzero elements. The constraint guarantees that only one codeword in each codebook can be activated to approximate the input data, hence it can lead to compact binary codes. As the binary constraints are directly imposed to the learning objective and are valid throughout the optimization procedure, the derived binary codes are much more accurate than sign thresholding binary codes. The rationale of using $M$ codebooks instead of single codebook to approximate each input datum is to further minimize quantization error, as the latter is shown to yield significantly lossy compression and incur evident performance drop \cite{cite:ICML14CQ,cite:CVPR14AQ}. Quantization based on multiple codebooks yields balanced composite binary codes which are more effective than Hamming embedding binary codes \cite{cite:TPAMI11PQ,cite:CVPR13CK}.

\subsubsection{Inter-Modality Correlation Maximization}
The most desirable value of multimodal retrieval is to enable transfer of knowledge across different modalities so that cross-modal retrieval performance can be improved. A fundamental assumption for multimodal retrieval is that by mapping objects in a modality-consistent latent space, the latent space representations of semantically relevant inter-modal pairs should be consistent. More specifically, for each input object with both image modality ${\bf x}^1_n$ and text modality ${\bf x}^2_n$, after being transformed by ${\bf R}^1$ and ${\bf R}^2$ in Equation~\eqref{eqn:CCQintra}, the latent space representations for image modality ${\bf C}^1 {\bf b}^1_n$ and text modality ${\bf C}^2 {\bf b}^2_n$ should be similar. To our knowledge, most prior work adopts the coupling strategy to minimize $\left\| {{{\mathbf{C}}^1}{\mathbf{b}}_n^1 - {{\mathbf{C}}^2}{\mathbf{b}}_n^2} \right\|_2^2$. In this paper, we propose to maximize cross-modal correlation by sharing codebooks $\{{\bf C}_m\}_{m=1}^M$ for different modalities and sharing binary codes $\{{\bf b}_n\}_{n=1}^{N_0}$ for semantically relevant inter-modal pairs. While for the data points with only one modality, the multimodal sharing strategy does not apply. Hence, the proposed condition that the modality-consistent latent space should satisfy is formulated as
\begin{equation}\label{eqn:CCQinter}
	{\bf{C}}_m^v = {{\bf{C}}_m} \\
	\;\;{\rm and} \;\;
  \delta \left( {{\bf{b}}_{{m}n}^{ v }} \right) = 
    \begin{cases}
      {{\bf{b}}_{{m}n}},&{n} = 1 \ldots {N_0} \\
      {\bf{b}}_{{m}n}^{ v },&{\rm{otherwise}}, \\ 
    \end{cases} \\
\end{equation}
where $\delta(\cdot)$ distinguishes multimodal objects from unimodal ones. Different from most prior methods \cite{cite:SIGMOD13IMH,cite:MM14CorrAE}, our modality-consistent condition requires identical code ${\bf{b}}_n^1 = {\bf{b}}_n^2$, instead of minimized distance $\left\| {{\bf{b}}_{n}^1 - {\bf{b}}_{n}^2} \right\|$, for the semantically relevant inter-modal pairs. There are two advantages of our approach. First, since our learning objective keeps the binary constraint valid throughout optimization procedure, it is very difficult to require minimized distance between two binary codes as their nonzero elements may differ significantly. Note that prior methods simply drop the binary condition and solve a continuous problem, which leads to uncontrollable quantization error with the post-step sign thresholding. Second, integrating the minimized distance condition in the learning objective as existing methods may introduce a trade-off term, or parameter, that is hard to tune since its magnitude is very different from  learning loss \eqref{eqn:CCQintra}.

\subsubsection{Joint Optimization Framework}
To approach CCQ, which jointly learns a similarity-preserving composite quantization and a correlation-maximal latent space in a unified optimization framework, we jointly require the codebooks $\{{\bf C}_m\}_{m=1}^M$ subject to minimizing the quantization error of all modalities as Equation~\eqref{eqn:CCQintra}, and the mappings ${\bf R}^v$ subject to maximizing the correlations between semantically relevant inter-modal pairs as Equation~\eqref{eqn:CCQinter} while jointly minimizing the reconstruction error of input data as LSA. This leads to a joint optimization framework as
\begin{equation}\label{eqn:CCQ}
  \begin{aligned}
    \mathop {\min }\limits_{{{\bf{R}}^v},{\bf{C}},{{\bf{B}}^v}} & \sum\limits_{v = 1}^V {\sum\limits_{n = 1}^{{N_v}} {{\lambda _v}\left\| {{\bf{x}}_n^v - {{\bf{R}}^v}\sum\limits_{m = 1}^M {{{\bf{C}}_m}\delta \left( {{\bf{b}}_{mn}^v} \right)} } \right\|_2^2} } \\
    {\rm{s.t.}} \quad & {{\bf{R}}^{v\mathsf{T}}}{{\bf{R}}^v} = {\bf{I}}_{D \times D}, {{\bf{R}}^{ v }} \in {\mathbb{R}^{{P_v} \times {D}}} \\
    & {\left\| {\delta \left( {{\bf{b}}_{mn}^v} \right)} \right\|_0} = 1,\delta\left( {\bf{b}}_{mn}^v \right) \in {\left\{ {0,1} \right\}^K} \\
    & \delta \left( {{\bf{b}}_{mn}^{ v }} \right) = 
      \begin{cases}
        {{\bf{b}}_{mn}},&{n} = 1 \ldots {N_0} \\
        {\bf{b}}_{mn}^{ v },&{\rm{otherwise}} \\ 
      \end{cases} \\
    & v = 1 \ldots V,m = 1 \ldots M,n = 1 \ldots N_v, \\ 
  \end{aligned}
\end{equation}
where $\lambda_v$ is the weight parameter for each modality, and in bimodal problems with $V=2$, we can simplify the notations by denoting $\lambda_1 = 1$ and $\lambda_2 = \lambda$, while such notations are used throughout this paper. ${\bf R}^v$ is the transformation matrix that maps the inputs of each modality to a $D$-dimensional modality-consistent latent space. The orthogonal constraints are motivated by LSA, which can turn latent factors ${\bf R}^v$ into transformation matrices for efficient out-of-sample quantization. The binary codes ${\bf b}^v_n$ are $M \times K$-dimensional, fortunately however, each ${\bf b}^v_{mn}$ is $1$-of-$K$ encoding with only one nonzero element and can be represented using ${\log _2}K$ bits, hence the final hash codes ${\bf b}^v_n$ can be compacted into $H = M{\log_2}K$ bits, which is independent on the dimensions of input or latent spaces. To fit each ${\bf b}^v_{mn}$ into one byte, $K = 256$ is a good choice \cite{cite:TPAMI11PQ,cite:ICML14CQ}. We simply set $D=\min(\{P_v\}_{v=1}^V,H)$, in the sense that an $H$-bit binary code can reconstruct a $D$-dimensional vector accurately.

Formally, we derive correlation-maximal mappings ${f^v}\left( {{\bf{x}}_n^v} \right) = {{\bf{R}}^{v\mathsf{T}}}{\bf{x}}_n^v$ and similarity-preserving quantizers ${q^v}\left( {{f^v}\left( {{\bf{x}}_n^v} \right)} \right) = {\bf{b}}_n^v$. There are several advantages of the CCQ approach. First, CCQ jointly learns a correlation-maximal latent space and a similarity-preserving composite encoding, which can minimize the quantization loss and guarantee search quality. Second, CCQ explores both paired and unpaired data in a semi-paired quantization paradigm, which can benefit from semi-supervised learning in that paired data consolidate inter-modality correlation and unpaired data enhance intra-modality quantization. Third, CCQ is formulated with only two easy-tuning model parameters $D$ and $\lambda$, where $D$ can be set as simply as LSA to retain most covariance information, and $\lambda$ can be selected by trading off different modalities using prior information. In particular, the proposed sharing of codebooks and binary codes across modalities \eqref{eqn:CCQinter} enables joint learning of latent semantics that are maximally correlated in the isomorphic feature space, which contributes most significantly to the efficacy of the CCQ approach.

\subsection{Approximate Nearest Neighbor Search}
Approximate nearest neighbor (ANN) search based on Euclidean distance is a powerful task for quantization techniques \cite{cite:TPAMI11PQ}. Given a database of CCQ hash codes $\{{\bf b}^{v}_n\}_{n=1}^{N_{v}}$, we follow \cite{cite:TPAMI11PQ,cite:CVPR13CK} and use \emph{Asymmetric Quantizer Distance} (AQD) as similarity metric that computes the distance between query ${\bf q}^{\bar v}$ and database point ${\bf x}^{v}_{n}$ as
\begin{equation}\label{eqn:AQD}
	\begin{aligned}
		{\text{AQD}} &\left( {{{\bf{q}}^{\bar v}},{{\bf{x}}^{v}_n}} \right) = \left\| {{{\bf{q}}^{\bar v}} - {{\bf{R}}^{\bar v}}\sum\nolimits_{m = 1}^M {{{\bf{C}}_m}{\bf{b}}_{mn}^{v}} } \right\|_2^2 \\
  	&= \textstyle{- 2\sum\nolimits_{m = 1}^M {\left\langle {{{\tilde{{\bf q}}}^{\bar v}},{{\bf{C}}_m}{\bf{b}}_{mn}^{v}} \right\rangle } + \left\| {\sum\nolimits_{m = 1}^M {{{\bf{C}}_m}{\bf{b}}_{mn}^{v}} } \right\|_2^2} \\
		&+ \left\| {{{\tilde{{\bf q}}}^{\bar v}}} \right\|_2^2 + \left\| {{{{{{\bf{R}}_{\bot}^{{\bar v}{\sf T}}}}}}{{\bf{{q}}}^{\bar v}}} \right\|_2^2, \\
	\end{aligned}
\end{equation}
where ${{\tilde{{\bf q}}}^{\bar v}} = {{\bf{R}}^{{\bar v}\mathsf{T}}}{{\bf{q}}^{\bar v}}$ is the transformed query. In the second row, the first term computes the inner products between ${\tilde {\bf q}}^{\bar v}$ and $M$ codewords selected by ${\bf b}^{v}_n$. Given a query, these inner products for all $M$ codebooks $\{{\bf C}_m\}_{m=1}^M$ and all $K$ possible values of ${\bf b}^{v}_{mn}$ can be pre-computed and stored in a query-specific $M \times K$ lookup table, which is used to compute AQD between the query and all database points, each entails $M$ table lookups and additions and is slightly more costly than Hamming distance. The second term computes the squared norm of decoded database point, which is independent on the query and can be encoded using one byte by quantizing these scale values on held-out dataset \cite{cite:CVPR14AQ}. At quantization, we augment CCQ code with the norm byte, which costs one more lookup and one more byte per database point. We can eliminate this norm byte by composite quantization \cite{cite:ICML14CQ}, but will leave it to our future work.

\section{Algorithm and Analysis}\label{section:Optimization}

\subsection{Learning Algorithm}
The CCQ optimization problem~\eqref{eqn:CCQ} consists of three variables, ${\bf R}^v$, ${\bf C}$, and ${\bf B}^v$. We adopt alternating optimization  \cite{cite:TPAMI11PQ,cite:CVPR13CK,cite:CVPR14AQ,cite:ICML14CQ} which iteratively updates one variable with the rest variables fixed.

\subsubsection{Update ${\bf R}^v$}
We update ${\bf R}^v$ by fixing ${\bf C}$ and ${\bf B}^v$ as known variables, and write Equation~\eqref{eqn:CCQ} with ${\bf R}^v$ as unknown variables in matrix formulation,
\begin{equation}\label{eqn:updateR}
	\begin{aligned}
  	\mathop {\min }\nolimits_{{{\bf{R}}^v}} &\left\| {{{\bf{X}}^v} - {{\bf{R}}^v}{\bf{C}}\delta \left( {{{\bf{B}}^v}} \right)} \right\|_F^2 \\
  	{\rm{s.t.}}\quad &{{\bf{R}}^{v{\mathsf{T}}}}{{\bf{R}}^v} = {{\bf{I}}_{D \times D}}. \\ 
	\end{aligned}
\end{equation}
This is equivalent to the Orthogonal Procrustes problem \cite{cite:P66OP} and can be solved exactly using SVD. More specifically, we perform SVD as ${{\bf{X}}^v}{\left[ {{\bf{C}}\delta \left( {{{\bf{B}}^v}} \right)} \right]}^{\mathsf{T}} = {\bf{US}}{{\bf{V}}^{\mathsf{T}}}$, then we achieve ${{\bf{R}}^v} = {\bf{U}}{{\bf{V}}^{\mathsf{T}}}$.

\subsubsection{Update ${\bf C}$}
We update ${\bf C}$ by fixing ${\bf R}^v$ and ${\bf B}^v$ as known variables, and write Equation~\eqref{eqn:CCQ} with ${\bf C}$ as unknown variables in matrix formulation,
\begin{equation}\label{eqn:updateC}
	\mathop {\min }\limits_{\bf{C}} \sum\nolimits_{v = 1}^V {\left\| {{{\bf{R}}^{v{\mathsf{T}}}}{{\bf{X}}^v} - {\bf{C}}\delta \left( {{{\bf{B}}^v}} \right)} \right\|_F^2}.
\end{equation}
This is an unconstrained quadratic problem with analytic solution ${\bf{C}} = \left[ {\sum\nolimits_{v = 1}^V {{\lambda _v}{{\bf{R}}^{v{\mathsf{T}}}}{{\bf{X}}^v}\delta \left( {{{\bf{B}}^v}} \right)^{\mathsf{T}}} } \right]{\left[ {\sum\nolimits_{v = 1}^V {{\lambda _v}\delta \left( {{{\bf{B}}^v}} \right) \delta {{\left( {{{\bf{B}}^v}} \right)}^{\mathsf{T}}}} } \right]^{ - 1}}$. Algorithms such as L-BFGS can be used to speed up computation.

\subsubsection{Update ${\bf B}^v$}
It is obvious that each ${\bf b}^v_n$ is independent on $\{{\bf b}^v_{n'}\}_{{n'} \ne n}$, then the optimization problem for ${\bf B}^v$ is decomposed to $N_v$ subproblems,
\begin{equation}\label{eqn:updateB}
  \begin{aligned}
    \mathop {\min }\limits_{{\bf{b}}_n^v} & \sum\nolimits_{v = 1}^V {{\lambda _v}\left\| {{{\bf{R}}^{v{\mathsf{T}}}}{\bf{x}}_n^v - \sum\nolimits_{m = 1}^M {{{\bf{C}}_m}\delta \left( {{\bf{b}}_{mn}^v} \right)} } \right\|_2^2}  \hfill \\
    & {\rm{s.t.}} \quad {\left\| {\delta \left( {{\bf{b}}_{mn}^v} \right)} \right\|_0} = 1,\delta\left( {\bf{b}}_{mn}^v \right) \in {\left\{ {0,1} \right\}^K}. \\
  \end{aligned}
\end{equation}
This optimization problem is generally NP-hard. As shown in \cite{cite:ICML14CQ}, this problem is essentially high-order Markov Random Field (MRF) problem and can be solved by the Iterated Conditional Modes (ICM) algorithm \cite{cite:JRSS86ICM} which solves $M$ indicators $\{{\bf b}^v_{mn}\}_{m=1}^M$ alternatively. Given $\{{\bf b}^v_{{m'}n}\}_{{m'} \ne m}$ fixed, we update ${\bf b}^v_{mn}$ by exhaustively checking all the codeword in codebook ${\bf C}_m$, finding the codeword such that the objective in \eqref{eqn:updateB} is minimized, and setting the corresponding entry of ${\bf b}^v_{mn}$ as $1$ and the rest as $0$. The algorithm is guaranteed to converge, and can be terminated if maximum iterations are reached. To accelerate quantization, we can explore hierarchical structure of codebooks $\{{\bf C}_m\}$ and update $\{{\bf b}^v_{mn}\}$ by a new greedy algorithm. Specifically, after updating $\{{\bf b}^v_{{m'}n}\}_{{m'}<{m}}$, we can update ${\bf b}^v_{mn}$ by encoding residual ${{\bf{R}}^{v{\mathsf{T}}}}{\bf{x}}_n^v - \sum\nolimits_{m' = 1}^{m - 1} {{{\bf{C}}_{m'}}\delta \left( {{\bf{b}}_{m'n}^v} \right)} $ with codebook ${\bf C}_m$.
The overall learning procedure is summarized in Algorithm~\ref{alg:CCQ}.

\begin{algorithm}[tbp]
  \DontPrintSemicolon
  \LinesNumbered
  \KwIn{Data $\{{\bf{X}}^v\}_{v=1}^V$; latent dimension $D$, modal weight $\lambda$.}
  \KwOut{Mappings $\{{\bf R}^v\}$, codebook ${\bf C}$, binary codes $\{{\bf B}^v\}$.}
  Initialize $\{{\bf R}^v\}$ by identity, ${\bf C}$ randomly, $\{{\bf B}^v\}$ by NN search.\;
  \Repeat{Convergence}{
		Update $\{{\bf R}^v\}$ by Orthogonal Procrustes as Eqn.~\eqref{eqn:updateR}.\;
		Update ${\bf C}$ by Quadratic Optimization as Eqn.~\eqref{eqn:updateC}.\;
		\For{$n \leftarrow 1$ \KwTo $N_v$}{
			Update $\{{\bf b}^v_n\}$ by ICM or greedy algorithm as Eqn.~\eqref{eqn:updateB}.\;
		}
	}
  \caption{CCQ: Composite Correlation Quantization} \label{alg:CCQ}
\end{algorithm}

\subsection{Large-Scale Implementation}
Batch algorithms are memory-inefficient for large-scale datasets, hence we formulate CCQ optimization into mini-batch algorithms for large-scale problems \cite{cite:VLDB14MSAE}. The main idea is to split the training set into mini-batches and load a fraction of data points into memory each time. Hence, the memory usage stays constant when the size of the training set increases. The update of ${\bf B}^v$ in Equation~\eqref{eqn:updateB} is already mini-batch in that update of each data point is independent on the other data points. To update ${\bf R}^v$ in mini-batch, we notice that the matrix for SVD is ${{\bf{X}}^v}{\left[ {{\bf{C}}\delta \left( {{{\bf{B}}^v}} \right)} \right]}^{\mathsf{T}} \in \mathbb{R}^{P_v \times D}$, which if given, the SVD can be solved in $O(P_v^2 D)$, independent on the number of data points. We thus formulate the matrix for SVD in a point-wise summation form as $\sum\nolimits_{n = 1}^{{N_v}} {{\bf{x}}_n^v{{\left[ {{\bf{C}}\delta \left( {{{\bf{b}}_n}} \right)} \right]}^{\mathsf{T}}}}$, then it can be computed by traversing all data points in a mini-batch paradigm. Similarly, the update of ${\bf C}$ can also be formulated in a summation form for mini-batch implementation. Note that we can allocate all available memory to mini-batch and trade off memory and disk reading costs.

\subsection{Computational Complexity}
We analyze the cost of each iteration to show CCQ scales linearly to sample size $N_v$. To update ${\bf R}^v$, it takes $O\left( {{N_v}{P_v}D + {N_v}DM} \right)$ to prepare the problem and $O\left( {P_v^2D + {D^3}} \right)$ to compute the SVD. To update ${\bf C}$, it takes $O\left( {{N_v}{P_v}D + {N_v}DM + {N_v}{M^2}} \right)$ to prepare the problem and $O\left( {D{M^2}{K^2} + {M^3}{K^3}} \right)$ to compute the quadratic optimization. To update ${\bf B}^v$, it takes $O\left( {{N_v}{P_v}D + {N_v}DMK{T_i}} \right)$, where $T_i$ is the number of iterations and $T_i = 3$ in ICM algorithm or $T_i = 1$ in greedy algorithm can obtain satisfactory performance. As a rule of thumb, $D = H$ and $K = 256$ are good choices for most applications. For longer codes, update of ${\bf C}$ is inefficient, in which case we can adopt the online L-BFGS algorithm for speedup.

\subsection{Approximation Error Analysis}
Given a query ${\bf q}^{\bar v}$ and a database point ${\bf x}^v_{n}$, after transformed by correlation-maximal mappings ${\tilde{{\bf q}}}^{\bar v} = {\bf R}^{\bar v\mathsf{T}} {\bf q}^{{\bar v}}$ and ${\tilde{\bf x}}^v_n = {\bf R}^{v\mathsf{T}} {\bf x}^v_{n}$, they can be comparable in the modality-consistent latent space, and their Euclidean distance is computed as $d\left( {{{\tilde{\bf{q}}}^{\bar v}},{\tilde{\bf{x}}}_n^v} \right) = {\left\| {{{{\tilde{\bf q}}}^{\bar v}} - {\tilde{\bf x}}_n^v} \right\|_2}$. As computing Euclidean distance on real-valued vectors is too costly for large-scale search, we compute AQD \eqref{eqn:AQD} on binary codes. Hence, we need to analyze the error bound of using AQD to approximate real-valued distance. Denote ${\hat{\bf x}}^v_n = {\sum\nolimits_{m = 1}^M {{{\bf{C}}_m}{\bf{b}}_{mn}^v} }$ the decoded vector of ${\bf x}^v_n$, then ${\text{AQD}}\left( {{{\bf{q}}^{\bar v}},{\bf{x}}_n^v} \right) = d\left( {{{\tilde{\bf{q}}}^{\bar v}},{\hat{\bf x}}_n^v} \right) + \epsilon$, $\epsilon$ is a constant.
\begin{theorem}[Bound]
	The error is bounded by learning loss
	\begin{equation}
		\textstyle
		\left| {d\left( {{{{\tilde{\bf q}}}^{\bar v}},{\tilde{\bf x}}_n^v} \right) - d\left( {{{{\tilde{\bf q}}}^{\bar v}},{\hat{\bf x}}_n^v} \right)} \right| \leqslant {\left\| {{\bf{x}}_n^v - {{\bf{R}}^v}\sum\nolimits_{m = 1}^M {{{\bf{C}}_m}{\bf{b}}_{mn}^v} } \right\|_2}.
	\end{equation}
\end{theorem}
\begin{proof}
From the triangle inequality, $ \left| {d\left( {{{{\tilde{\bf q}}}^{\bar v}},{\tilde{\bf x}}_n^v} \right) - d\left( {{{{\tilde{\bf q}}}^{\bar v}},{\hat{\bf x}}_n^v} \right)} \right| \leqslant d\left( {{\tilde{\bf x}}_n^v,{\hat{\bf x}}_n^v} \right) $. Then
\begin{equation}\label{eqn:boundapprox}
	\begin{aligned}
		d^2\left( {{\tilde{\bf x}}_n^v,{\hat{\bf x}}_n^v} \right) & = \textstyle{{\left\| {{{\bf{R}}^{v{\mathsf{T}}}}{\bf{x}}_n^v - \sum\nolimits_{m = 1}^M {{{\bf{C}}_m}{\bf{b}}_{mn}^v} } \right\|^2_2}} \\
		& \leqslant \textstyle{{\left\| {{{\bf{R}}^{v{\mathsf{T}}}}{\bf{x}}_n^v - \sum\nolimits_{m = 1}^M {{{\bf{C}}_m}{\bf{b}}_{mn}^v} } \right\|^2_2} + {\left\| {{{\bf{R}}_{ \bot }^{v{\sf T}}}{\bf{x}}_n^v} \right\|^2_2}} \\
		& = \textstyle{{\left\| {{\bf{x}}_n^v - {{\bf{R}}^v}\sum\nolimits_{m = 1}^M {{{\bf{C}}_m}{\bf{b}}_{mn}^v} } \right\|^2_2}},
	\end{aligned}
\end{equation}
where ${\bf R}_{ \bot }^v$ is an orthogonal complement of ${\bf{R}}^v$, ${{{\bf{R}}^{v{\sf T}}}{\bf{R}}_ \bot ^v} = {\bf{0}}$.
\end{proof}
The theorem confirms that the error of using AQD to approximate real-valued distance is statistically bounded by CCQ learning loss. Hence, CCQ is more accurate than sign thresholding methods \cite{cite:NIPS09SH}. An important advantage of CCQ in Equation~\eqref{eqn:boundapprox} is that mapping ${\bf R}^v$ is learned by a joint optimization of canonical correlation analysis (CCA) and principal component analysis (PCA) corresponding to the first and second terms of Line 2 in Equation~\eqref{eqn:boundapprox}. This can be much more effective than most CCA-based methods \cite{cite:IJCAI11CVH,cite:AAAI14SCM,cite:IJCAI15QCH}.

\section{Experiments}\label{section:Experiments}
We conduct extensive evaluation of CCQ against state of the art methods on three public multimodal datasets. We investigate both effectiveness and efficiency in terms of search precision, recall, and time. The codes, data, and configurations will be available online.

\subsection{Datasets}
The evaluation is conducted on three datasets: NUS-WIDE \cite{cite:CIVR09NusWide}, Wiki \cite{cite:TPAMI14Wiki}, and Flickr1M \cite{cite:ICMR08Flickr1M}, with statistics depicted in Table~\ref{table:datasets}. We preprocess all datasets by applying ZCA \cite{cite:VLDB14MSAE} to normalize each dimension of image/text features to be zero mean and unit variance.

\textbf{NUS-WIDE}\footnote{\scriptsize\url{http://lms.comp.nus.edu.sg/research/NUS-WIDE.htm}} is a Web image dataset containing 269, 648 images downloaded from Flickr, each associated with 6 tags on average.  There are 81 ground truth concepts manually annotated for search evaluation. Following prior works \cite{cite:MM13LCMH,cite:VLDB14MSAE}, we prune the original NUS-WIDE to form a new dataset consisting of 195,834 image-text pairs by keeping the pairs that belong to one of the 21 most frequent concepts. The images are represented by 500-dimensional bag-of-words vectors extracted from the SIFT features using k-means, and the texts are represented by 1,000-dimensional vectors extracted from the tag occurrence features using PCA. A query set of 2,000 image-text pairs are randomly sampled from the dataset, while the remaining 193,834 image-text pairs are serving as the database. The hash models are learned on the training set containing 10,000 image-text pairs randomly sampled from the database \cite{cite:MM13LCMH,cite:SIGMOD13IMH}.

\textbf{Wiki}\footnote{\scriptsize\url{http://www.svcl.ucsd.edu/projects/crossmodal}} contains 2,866 image-text pairs selected from Wikipedia's featured articles comprised of multiple sections of images and texts. Every image-text pair is labeled by one of the 10 concepts in the article categories. Each image is represented by a 128-dimensional bag-of-words vector extracted from SIFT features, and each text is represented by the probability distribution over 10 topics learned by a latent Dirichlet allocation (LDA) model. The dataset is released with a query set of 693 pairs and a database of 2,173 pairs, and the whole database is used as the training set for hash coding \cite{cite:TPAMI14Wiki,cite:MM13LCMH}.

\textbf{Flickr1M} comprises 1,000,000 images associated with tags from Flickr, in which 25,000 are labeled with 38 concepts while the remaining 975,000 are unlabeled. The public available preprocessed dataset\footnote{\scriptsize\url{http://www.cs.toronto.edu/~nitish/multimodal}} is employed for evaluation, in which each image is represented by a 3,857-dimensional vector concatenated by local SIFT feature, global GIST feature, etc \cite{cite:JMLR14MMDL}. Each text is represented by a 2,000-dimensional vector extracted from tag occurrences. The query set contains 1,000 image-text pairs randomly sampled from the 25,000 labeled pairs, and the rest 24,000 labeled pairs are used as the database. In scalability test of CCQ (Section \ref{section:scalability}), all 975,000 unlabeled pairs are used as the training set for learning hash codes.

\begin{table}[tbp]
    \addtolength{\tabcolsep}{6pt}
    \centering
    \caption{The Statistics of Three Datasets}
    \label{table:datasets}
    \begin{tabular}{|c|c|c|c|}
        \hline
        {Dataset} & {NUS-WIDE} & {Wiki} & {Flickr1M}\\
        \hline
        \hline
        Complete Set & 195,834 & 2,866 & 1,000,000 \\
        \hline
        Labeled Set & 195,834 & 2,866 & 25,000 \\
        \hline
        Query Set & 2,000 & 693 & 1,000 \\
        \hline
        Database & 193,834 & 2,173 & 24,000 \\
        \hline
        Training Set & 10,000 & 2,173 & 975,000 \\
        \hline
    \end{tabular}
\end{table}

\begin{figure*}[!htb]
    \centering
    \subfigure[${I \rightarrow T}$ @ 16 bits]{
        \includegraphics[width=0.23\textwidth]{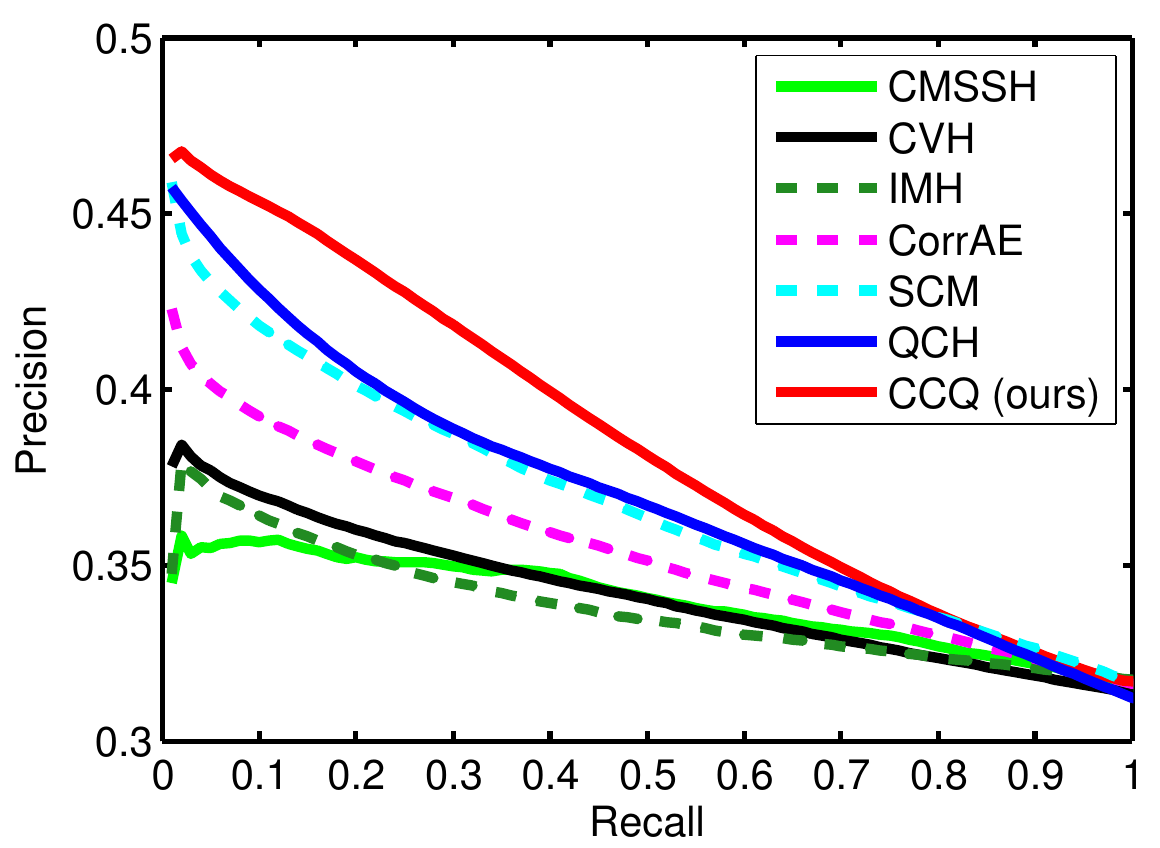}
        \label{fig:pr_nuswide_1}
    }
    \subfigure[${I \rightarrow T}$ @ 32 bits]{
        \includegraphics[width=0.23\textwidth]{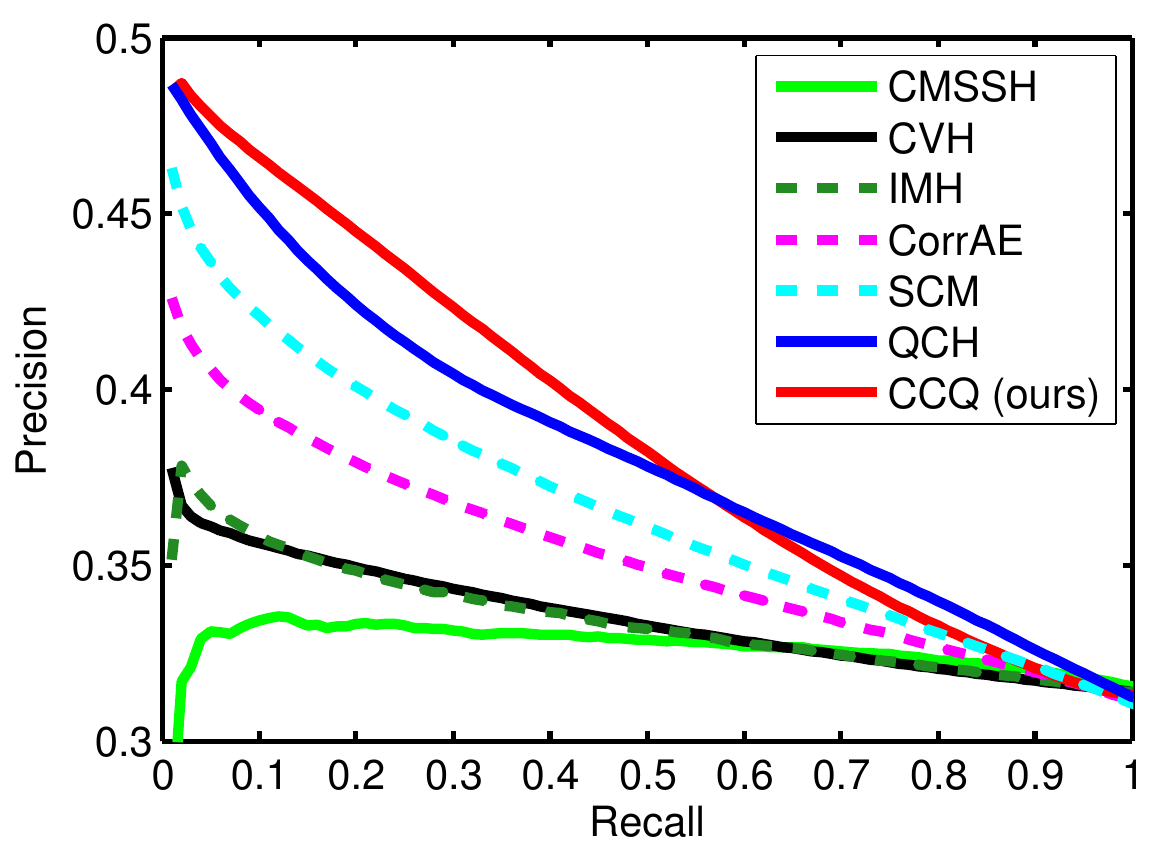}
        \label{fig:pr_nuswide_2}
    }
    \subfigure[${T \rightarrow I}$ @ 16 bits]{
        \includegraphics[width=0.23\textwidth]{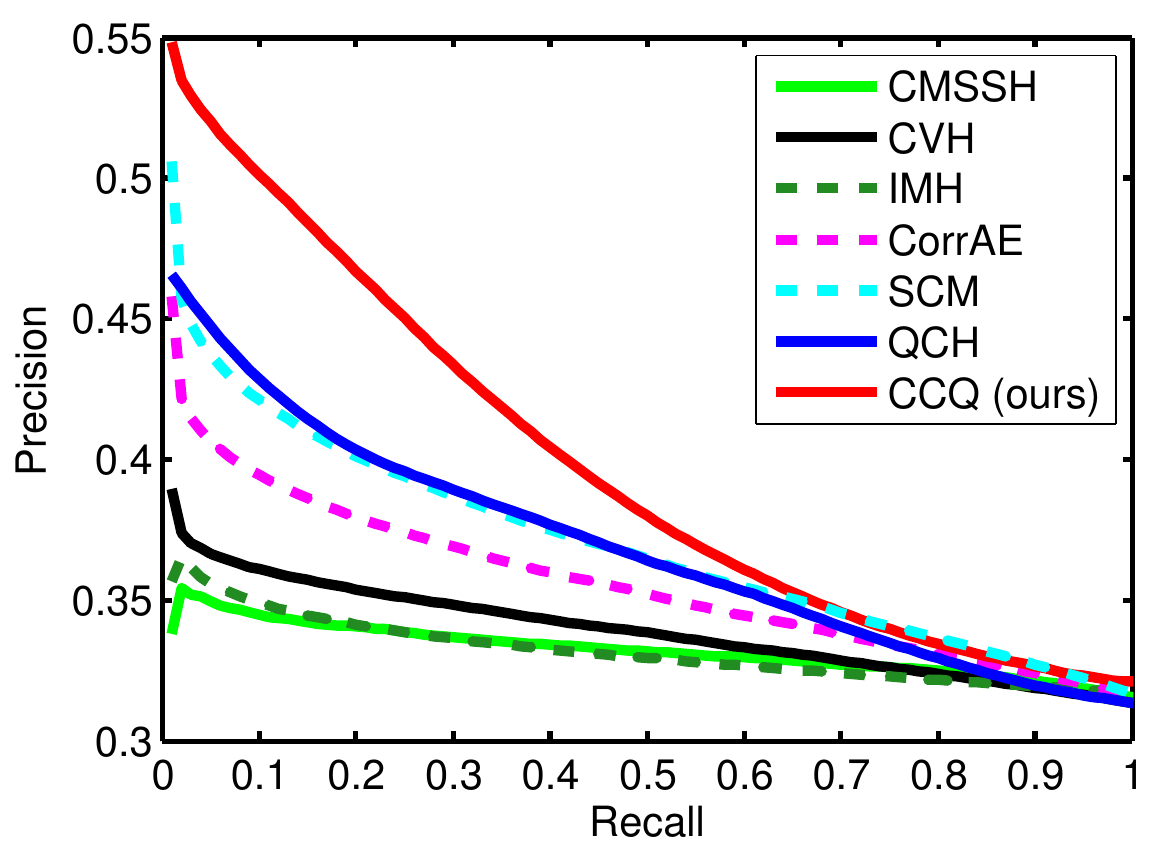}
        \label{fig:pr_nuswide_3}
    }
    \subfigure[${T \rightarrow I}$ @ 32 bits]{
        \includegraphics[width=0.23\textwidth]{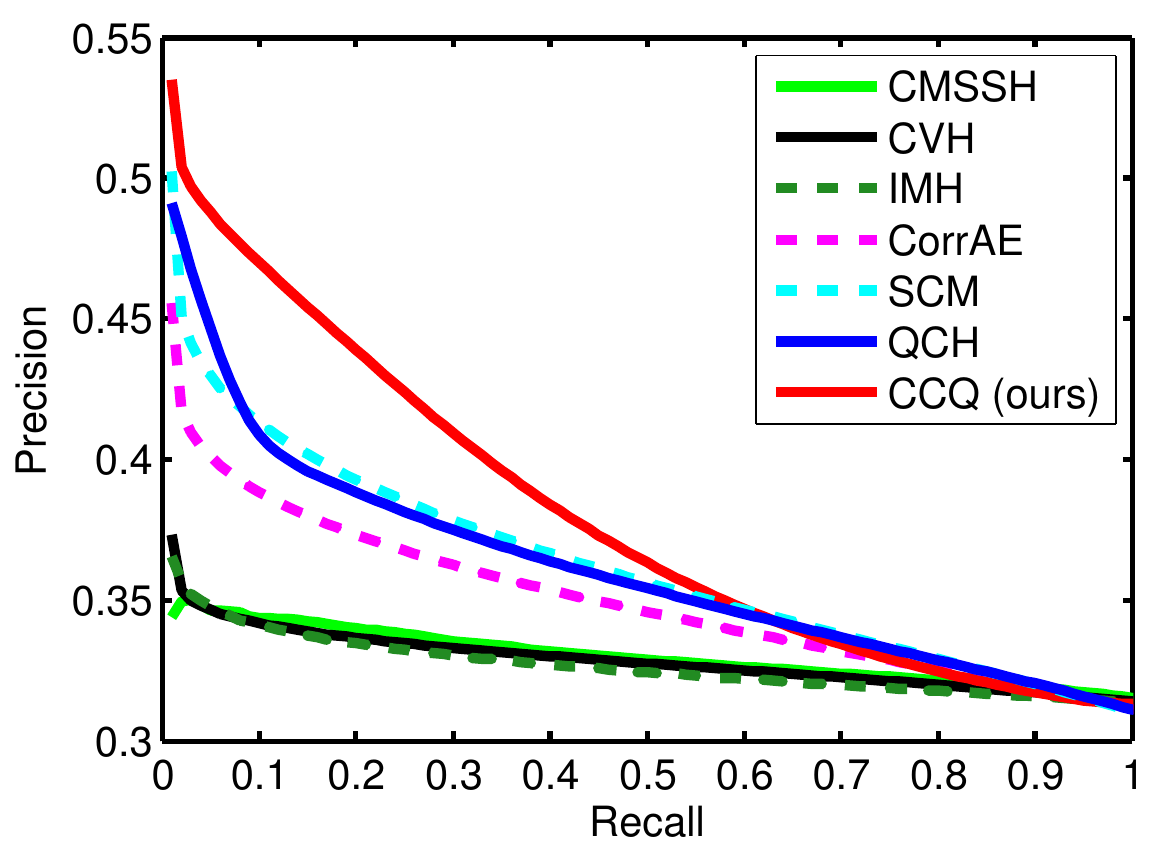}
        \label{fig:pr_nuswide_4}
    }
    \\
    \vspace{-5pt}
    \subfigure[${I \rightarrow T}$ @ 16 bits]{
        \includegraphics[width=0.23\textwidth]{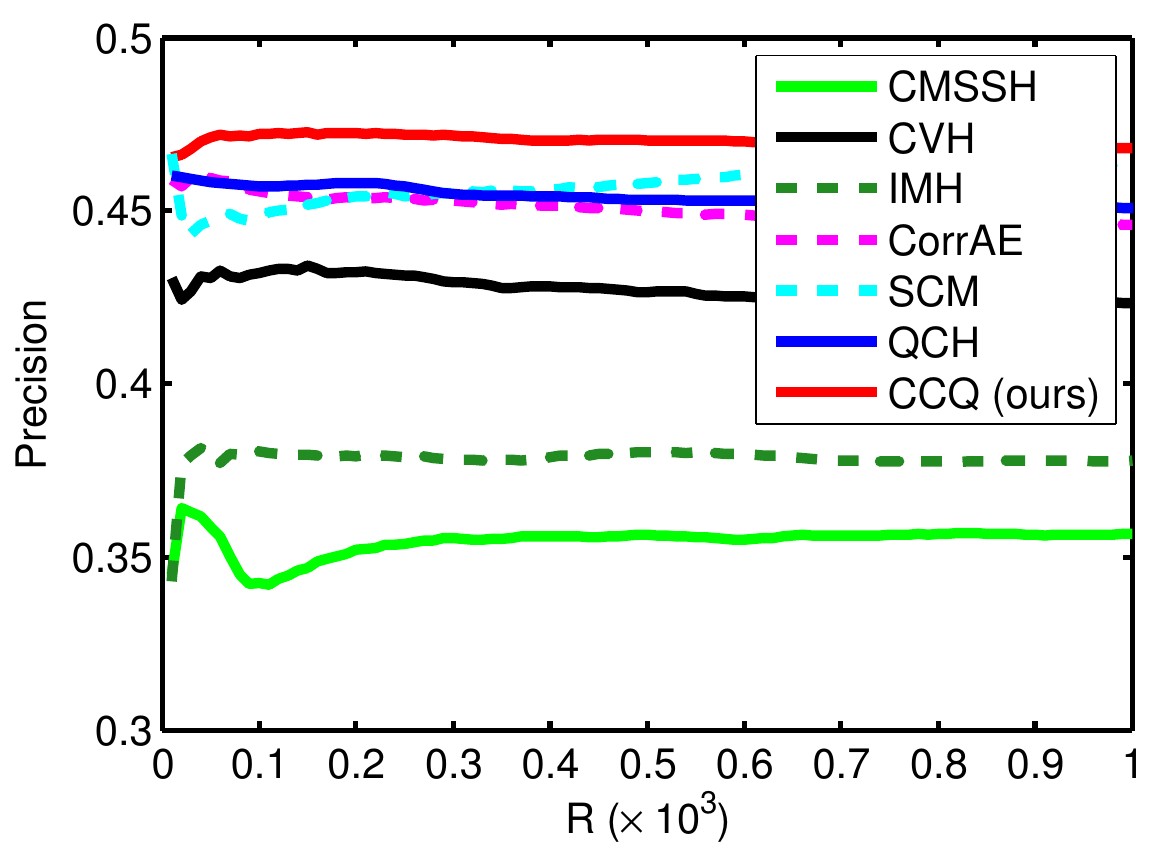}
        \label{fig:prec_nuswide_1}
    }
    \subfigure[${I \rightarrow T}$ @ 32 bits]{
        \includegraphics[width=0.23\textwidth]{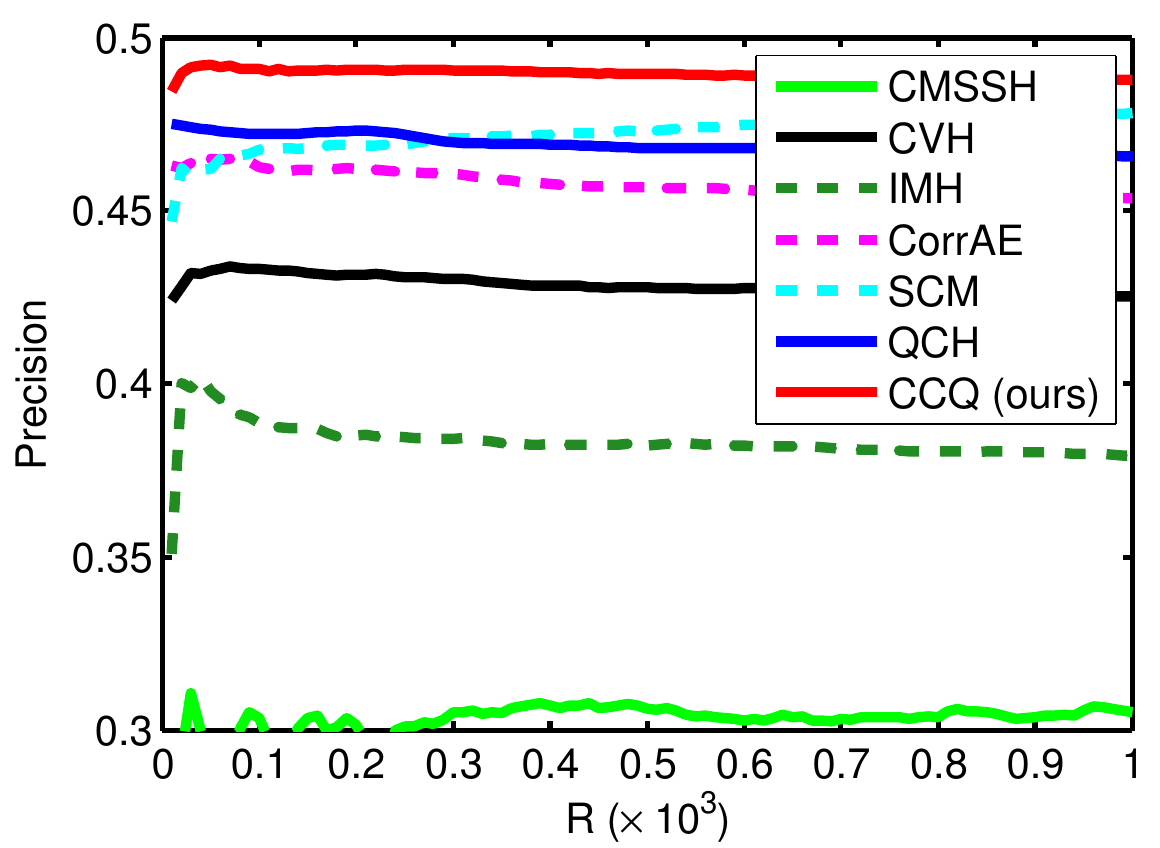}
        \label{fig:prec_nuswide_2}
    }
    \subfigure[${T \rightarrow I}$ @ 16 bits]{
        \includegraphics[width=0.23\textwidth]{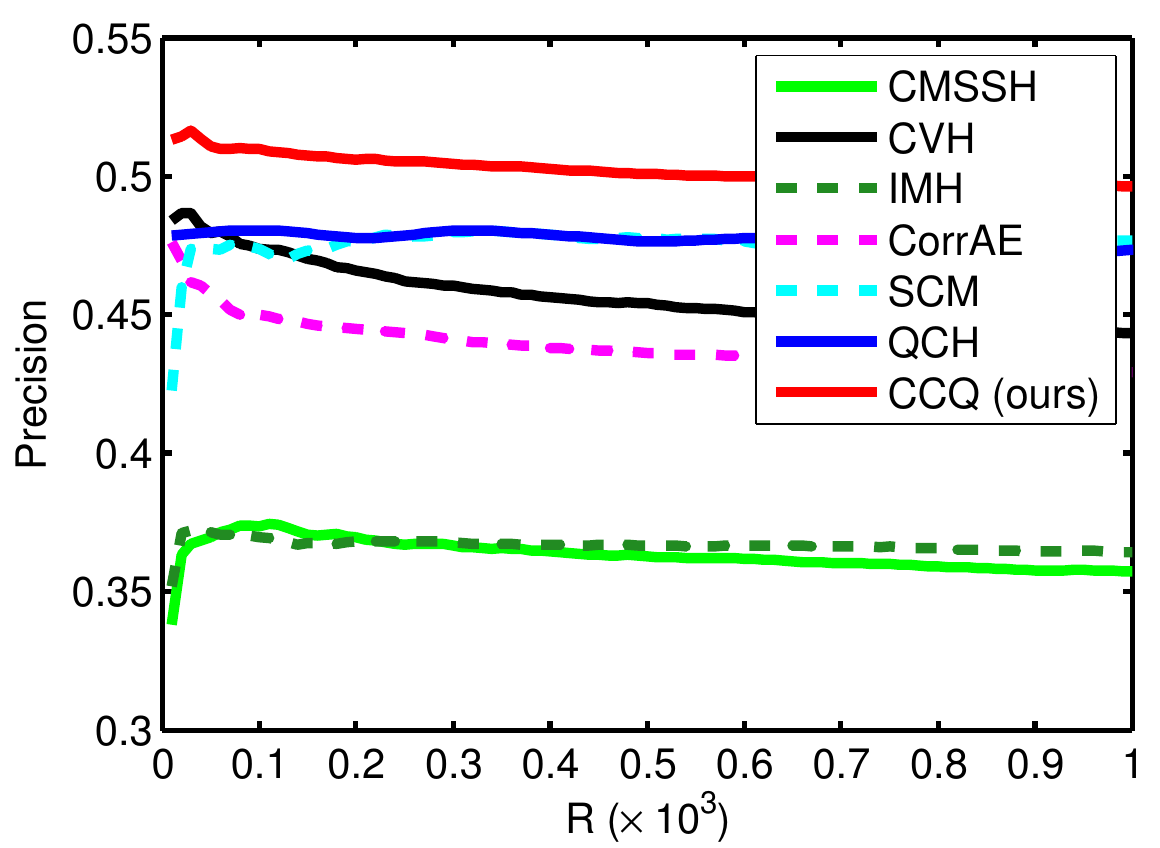}
        \label{fig:prec_nuswide_3}
    }
    \subfigure[${T \rightarrow I}$ @ 32 bits]{
        \includegraphics[width=0.23\textwidth]{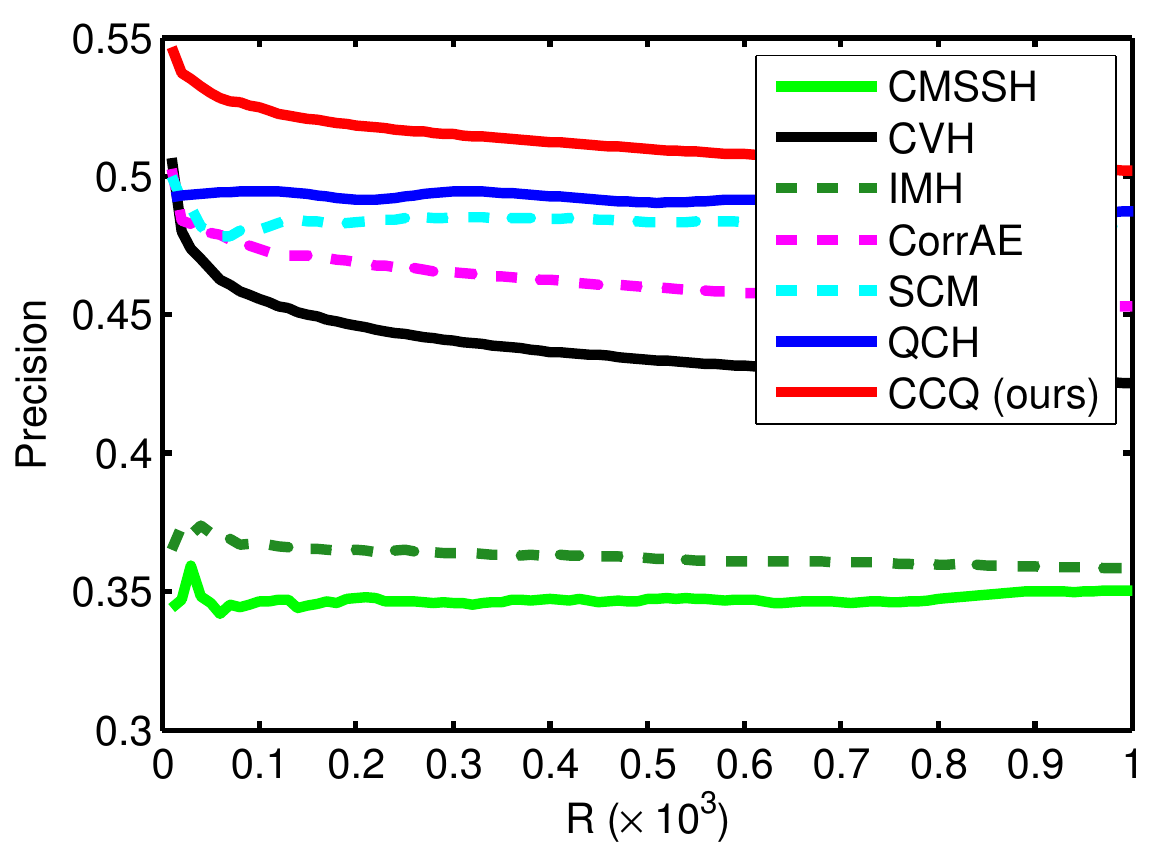}
        \label{fig:prec_nuswide_4}
    }
    \vspace{-10pt}
    \caption{Precision-recall curves (top) and precision@R curves (bottom) on NUS-WIDE cross-modal search tasks @ 16 and 32 bits.}
    \label{fig:nuswide}
\end{figure*}

\subsection{Comparison Methods} 
We compare CCQ against many state of the art hashing methods.
\begin{itemize}
	\setlength\itemsep{0pt}
	\item \textbf{Unsupervised hashing:} Cross-View Hashing (\textbf{CVH})\footref{foot:CMSSH} \cite{cite:IJCAI11CVH} and Inter-Media Hashing (\textbf{IMH})\footnote{\scriptsize\url{http://staff.itee.uq.edu.au/shenht/UQ_IMH}} \cite{cite:SIGMOD13IMH} are unsupervised hashing methods that extend spectral hashing to exploit the local structure of multimodal data for learning binary codes.
	\item \textbf{Deep hashing:} Correspondence Auto-Encoders (\textbf{CorrAE})\footnote{\scriptsize\url{https://github.com/fangxiangfeng/deepnet}} \cite{cite:MM14CorrAE} learns latent features via unsupervised deep auto-encoders, which captures both intra-modal and inter-modal correspondences, and binarizes latent features via sign thresholding.
	\item \textbf{Supervised hashing:} Cross-Modal Similarity-Sensitive Hashing (\textbf{CMSSH})\footnote{\scriptsize\url{http://www.cse.ust.hk/~dyyeung/code/mlbe.zip}\label{foot:CMSSH}} \cite{cite:CVPR10CMSSH}, Semantic Correlation Maximization (\textbf{SCM}) \cite{cite:AAAI14SCM}, and Quantized Correlation Hashing (\textbf{QCH}) are supervised hashing methods which embed multimodal data into a common Hamming space using supervised metric learning.
\end{itemize}

\subsection{Evaluation Protocols}
We perform four types of multimodal retrieval schemes: (1) ${I \rightarrow I}$: use image queries to retrieve relevant images; (2) ${T \rightarrow T}$: use text queries to retrieve relevant texts; (3) ${I \rightarrow T}$: use image queries to retrieve relevant texts; and (4) ${T \rightarrow I}$: use text queries to retrieve relevant images. The first two tasks are intra-modal retrieval and the last two tasks are cross-modal retrieval. As CCQ can also handle multimodal search where both modalities are available for the database, we show the results of multimodal retrieval schemes where each image-text pair is quantized into a unified hash code by fusing knowledge of different modalities: (5) $I \rightarrow IT$: use image queries to retrieve relevant image-text pairs; (6) $T \rightarrow IT$: use text queries to retrieve relevant image-text pairs. The baseline methods do not support multimodal search because they do not use shared coding for different modalities of the same object. Given a query, the ground truth is defined as: if a result shares at least one common concept with the query, it is relevant; otherwise it is irrelevant.

We adopt \emph{Mean Average Precision} (MAP) to measure the effectiveness of multimodal search \cite{cite:SIGMOD13IMH,cite:MM13LCMH,cite:VLDB14MSAE,cite:SIGIR14DCDH,cite:MM14CorrAE}. Given a set of queries, we first calculate Average Precision (AP) of each query as
\begin{equation}
	{\text{AP}}\text{@}R = \frac{{\sum\nolimits_{r = 1}^R {P\left( r \right)\delta \left( r \right)} }}{{\sum\nolimits_{r' = 1}^R {\delta \left( {r'} \right)} }},
\end{equation}
where $R$ is the number of retrieved documents, $P(r)$ denotes the precision of the top $r$ retrieved results, and $\delta(r) = 1$ if the $r$-th retrieved result is a true neighbor of the query, otherwise $\delta(r) = 0$. Then MAP is computed as the mean of all the queries' average precision, and the larger the MAP, the better the retrieval performance. In the experiments, we follow \cite{cite:SIGIR13CMR,cite:SIGIR14DCDH,cite:VLDB14MSAE} to report MAP@$R=50$. We also report another two standard retrieval criteria, \emph{precision-recall} curves and \emph{precision@top-$R$} curves of all retrieval tasks. In addition to effectiveness, we report \emph{time and memory} costs as the efficiency measures for query processing and model training.

The CCQ approach involves two model parameters: dimension of modality-consistent subspace $D$ and modality trade-off weight $\lambda$. In principle, CCQ is almost immune to different choices of $D$, as long as $D$ is large enough to retain the majority amount of covariance information as LSA. While no prior knowledge is available, we can simply set equal weights $\lambda = 1$ for different modalities, which can already achieve satisfactory performance. Nonetheless, for image-text bimodal search, the text modality usually carry more semantic information, hence we equip CCQ with the flexibility for selecting the optimal $\lambda$ to encode such important prior knowledge. Given annotation ground truths as in the evaluation datasets, we can automatically select $D$ and $\lambda$ using cross-validation. However, we choose to blindly fix $\lambda = 5$ throughout the comparative study. This is desirable as cross-validation may be impossible in the pervasive unsupervised multimodal search. We will study parameter sensitivity in Section~\ref{section:Sensitivity} to validate that CCQ can consistently outperform the state of the arts with a wide range of parameter configurations.

For the comparison methods, we adopt cross-validation to select their optimal parameters, respectively. As cross-validation requires annotation ground truths, this further confirms CCQ's superior parameter stability. Subject to computation burden, it is too costly to train CMSSH and IMH on the complete Flickr1M dataset, hence we randomly sample 10,000 image-text pairs to train these models. Each experiment repeats ten runs and the average result is reported.

\begin{table*}[!htb]
    \addtolength{\tabcolsep}{0.6pt}
    \centering 
    \caption{Mean Average Precision (MAP) Comparison of Six Multimodal Retrieval Tasks on Three Standard Datasets}
    \label{table:MAP}
    \small
    \begin{tabular}{|c|c|cccc|cccc|cccc|}
        \hline
        \multirow{2}{25pt}{\centering Task} & \multirow{2}{25pt}{\centering Method} & \multicolumn{4}{c|}{NUS-WIDE} & \multicolumn{4}{c|}{Wiki} & \multicolumn{4}{c|}{Flickr1M}\\
        \cline{3-14}
        & & 8 bits & 16 bits  & 32 bits  & 64 bits  & 8 bits  & 16 bits  & 32 bits  & 64 bits  & 8 bits & 16 bits  & 32 bits  & 64 bits \\
        \hline
        \hline
        \multirow{7}{25pt}{\centering ${I \rightarrow I}$} 
        & CVH \cite{cite:IJCAI11CVH} & 0.3954 & 0.4542 & 0.4759 & 0.4780 & 0.1988 & 0.1969 & 0.2042 & 0.2058 & 0.6050 & 0.6328 & 0.6615 & 0.6712 \\
        & IMH \cite{cite:SIGMOD13IMH} & 0.4313 & 0.4545 & 0.4155 & 0.4005 & 0.1910 & 0.1963 & 0.1937 & 0.1935 & 0.5239 & 0.5725 & 0.5736 & 0.5748 \\
        & CorrAE \cite{cite:MM14CorrAE} & 0.4223 & 0.4478 & 0.4587 & 0.4796 & 0.2055 & 0.2086 & 0.2188 & 0.2194 & 0.6145 & 0.6397 & 0.6588 & 0.6654 \\
        \cline{2-14}
        & CMSSH \cite{cite:CVPR10CMSSH} & 0.3776 & 0.4060 & 0.4356 & 0.4490 & 0.1987 & 0.1979 & 0.2007 & 0.2126 & 0.5738 & 0.6304 & 0.6587 & 0.6932 \\
        & SCM \cite{cite:AAAI14SCM} & 0.4258 & 0.4578 & 0.4695 & 0.4831 & 0.2048 & 0.2103 & 0.2177 & 0.2212 & 0.5926 & 0.6257 & 0.6615 & 0.6801 \\
        & QCH \cite{cite:IJCAI15QCH} & 0.4289  & 0.4557 & 0.4786 & 0.4898 & 0.2087 & 0.2155 & 0.2198 & 0.2252 & 0.6165 & 0.6586 & 0.6787 & 0.6885 \\
        \cline{2-14}
        & CCQ (ours) & \textbf{0.4711} & \textbf{0.4859} & \textbf{0.4921} & \textbf{0.4932} & \textbf{0.2226} & \textbf{0.2265} & \textbf{0.2373} & \textbf{0.2386} & \textbf{0.6714} & \textbf{0.7092} & \textbf{0.7318} & \textbf{0.7451} \\
        \hline
        \hline
        \multirow{7}{25pt}{\centering ${T \rightarrow T}$} 
        & CVH \cite{cite:IJCAI11CVH} & 0.5825 & 0.6485 & 0.6837 & \textbf{0.7189} & 0.4049 & 0.5506 & 0.6075 & 0.6239 & 0.5812 & 0.6085 & 0.6242 & 0.6337 \\
        & IMH \cite{cite:SIGMOD13IMH} & 0.4531 & 0.4740 & 0.5421 & 0.6202 & 0.3805 & 0.4623 & 0.5773 & 0.5989 & 0.5585 & 0.5973 & 0.6360 & 0.6436 \\
        & CorrAE \cite{cite:MM14CorrAE} & 0.5501 & 0.5856 & 0.6344 & 0.6678 & 0.5765 & 0.5889 & 0.6045 & 0.6123 & 0.6060 & 0.6176 & 0.6389 & 0.6443 \\
        \cline{2-14}
				& CMSSH \cite{cite:CVPR10CMSSH} & \textbf{0.5911} & 0.5968 & 0.6215 & 0.6613 & 0.5503 & 0.6065 & 0.6188 & 0.6232 & 0.5487 & 0.5573 & 0.5583 & 0.5614 \\
        & SCM \cite{cite:AAAI14SCM} & 0.5524 & 0.6315 & 0.6606 & 0.6736 & 0.5814 & 0.6051 & 0.6189 & 0.6324 & 0.5924 & 0.6320 & 0.6410 & 0.6485 \\
        & QCH \cite{cite:IJCAI15QCH} & 0.5706 & \textbf{0.6586} & 0.6796 & 0.6855 & 0.6002 & 0.6128 & 0.6226 & 0.6355 & 0.6022 & 0.6427 & \textbf{0.6554} & \textbf{0.6686} \\
        \cline{2-14}
        & CCQ (ours) & \textbf{0.5913} & 0.6481 & \textbf{0.6917} & 0.7069 & \textbf{0.6017} & \textbf{0.6286} & \textbf{0.6366} & \textbf{0.6422} & \textbf{0.6090} & \textbf{0.6433} & 0.6541 & 0.6550 \\
        \hline
        \hline
        \multirow{7}{25pt}{\centering ${I \rightarrow T}$} 
        & CVH \cite{cite:IJCAI11CVH} & 0.4588 & 0.4713 & 0.4743 & 0.4740 & 0.1673 & 0.1877 & 0.1716 & 0.1696 & 0.6091 & 0.6225 & 0.6364 & 0.6199 \\
        & IMH \cite{cite:SIGMOD13IMH} & 0.4345 & 0.4399 & 0.4203 & 0.4115 & 0.1734 & 0.1896 & 0.1714 & 0.1601 & 0.5449 & 0.5646 & 0.5936 & 0.5539 \\
        & CorrAE \cite{cite:MM14CorrAE} & 0.4398 & 0.4522 & 0.4699 & 0.4964 & 0.1929 & 0.1982 & 0.2033 & 0.2155 & 0.6301 & 0.6329 & 0.6357 & 0.6401 \\
        \cline{2-14}
        & CMSSH \cite{cite:CVPR10CMSSH} & 0.3950 & 0.4052 & 0.4076 & 0.3516 & 0.1672 & 0.1727 & 0.1750 & 0.1759 & 0.5076 & 0.5272 & 0.5357 & 0.5219 \\
        & SCM \cite{cite:AAAI14SCM} & 0.4693 & 0.4648 & 0.4619 & 0.4851 & 0.2258 & 0.2372 & 0.2381 & 0.2378 & 0.6361 & 0.6493 & 0.6495 & 0.6440 \\
        & QCH \cite{cite:IJCAI15QCH} & 0.4765 & 0.4895 & 0.5050 & 0.5125 & 0.2288 & 0.2343 & 0.2368 & \textbf{0.2402} & 0.6452 & 0.6523 & 0.6685 & 0.6721 \\
        \cline{2-14}
        & CCQ (ours) & \textbf{0.5124} & \textbf{0.5161} & \textbf{0.5165} & \textbf{0.5372} & \textbf{0.2338} & \textbf{0.2349} & \textbf{0.2371} & 0.2374 & \textbf{0.6879} & \textbf{0.7081} & \textbf{0.7183} & \textbf{0.7176} \\
        \hline 
        \hline 
        {\centering ${I \rightarrow IT}$} & CCQ (ours) & 0.5074 & 0.5411 & 0.5414 & 0.5441 & 0.2512 & 0.2513 & 0.2529 & 0.2587 & 0.7063 & 0.6894 & 0.6989 & 0.6996\\
        \hline 
        \hline
        \multirow{7}{25pt}{\centering ${T \rightarrow I}$} 
        & CVH \cite{cite:IJCAI11CVH} & \textbf{0.5598} & 0.5217 & 0.5129 & 0.4875 & 0.2309 & 0.2219 & 0.2214 & 0.2350 & 0.5972 & 0.6032 & 0.5738 & 0.5794 \\
        & IMH \cite{cite:SIGMOD13IMH} & 0.4380 & 0.4582 & 0.4186 & 0.4051 & 0.2394 & 0.2227 & 0.2333 & 0.1896 & 0.5374 & 0.5536 & 0.5513 & 0.5583 \\
        & CorrAE \cite{cite:MM14CorrAE} & 0.4303 & 0.4501 & 0.4634 & 0.4880 & 0.2688 & 0.2928 & 0.3478 & 0.3566 & 0.6142 & 0.6198 & 0.6247 & 0.6431 \\
        \cline{2-14}
        & CMSSH \cite{cite:CVPR10CMSSH} & 0.3783 & 0.3499 & 0.3944 & 0.4015 & 0.2926 & 0.2991 & 0.2537 & 0.2582 & 0.5868 & 0.5732 & 0.6176 & 0.6323 \\
        & SCM \cite{cite:AAAI14SCM} & 0.4449 & 0.4859 & 0.5105 & 0.5259 & 0.3157 & 0.3698 & 0.4239 & 0.4369 & 0.6037 & 0.5998 & 0.5805 & 0.6078 \\
        & QCH \cite{cite:IJCAI15QCH} & 0.5020 & 0.5195 & \textbf{0.5489} & \textbf{0.5622} & 0.3426 & 0.3753 & \textbf{0.4411} & \textbf{0.4565} & 0.6258 & 0.6425 & 0.6485 & 0.6528 \\
        \cline{2-14}
        & CCQ (ours) & 0.5359 & \textbf{0.5410} & \textbf{0.5413} & 0.5556 & \textbf{0.3885} & \textbf{0.4000} & 0.4222 & 0.4178 & \textbf{0.6548} & \textbf{0.7026} & \textbf{0.7165} & \textbf{0.7266} \\
        \hline 
        \hline 
        {\centering ${T \rightarrow IT}$} & CCQ (ours) & 0.6022 & 0.6925 & 0.7131 & 0.7153 & 0.6355 & 0.6351 & 0.6394 & 0.6405 & 0.6942 & 0.7151 & 0.7190 & 0.7416 \\
        \hline 
    \end{tabular}
    \normalsize
\end{table*}

\subsection{Experimental Results}
We compare CCQ with state of the art methods in terms of MAP and precision-recall on 4 multimodal retrieval tasks (${I \rightarrow I}$, ${T \rightarrow T}$, ${I \rightarrow T}$, ${T \rightarrow I}$) of three datasets (NUS-WIDE, Wiki, and Flickr1M).

\subsubsection{Results on NUS-WIDE}
We evaluate CCQ against state of the arts with different lengths of hash codes, i.e. 8, 16, 32, and 64 bits, and report the MAP results in Table~\ref{table:MAP}. For all multimodal retrieval tasks, CCQ achieves significantly better performance than all unsupervised hashing methods CVH, IMH, and CorrAE, and generally outperforms the state of the art supervised hashing methods CMSSH, SCM, QCH in most cases. It is very worth noting that, CCQ is an \emph{unsupervised} hashing method that does not require labeled similarity information. Hence CCQ is particularly beneficial when labeled information is unavailable, which is the most common scenario in big data era. A notable limitation of orthogonal constrained methods CVH and IMH is that longer codes do not necessarily improve performance in cross-modal tasks ${I \rightarrow T}$ and ${T \rightarrow I}$. The reason is that these methods learn uncorrelated hash bits via eigenvalue decomposition on similarity matrix, which leads to unbalanced hash codes with the first $k$ eigenvectors (hash bits) dominating the whole hash codes. CCQ via composite quantization in isomorphic space can learn balanced binary codes, hence its performance improves with longer codes.

It is interesting to observe that the performances of cross-modal search task ${I \rightarrow T}$ is generally better than that of intra-modal search task ${I \rightarrow I}$, while this observation does not hold for the counterparts ${T \rightarrow I}$ and ${T \rightarrow T}$. This seems abnormal at first sight as cross-modal search tasks are often more challenging than intra-modal search tasks due to semantic gap \cite{cite:TPAMI14Wiki}. However, in general, text retrieval is much easier than image retrieval, making different modalities of the objects contribute differently the cross-modal retrieval performance. We believe that ${T \rightarrow T}$ is much easier than ${T \rightarrow I}$, but ${I \rightarrow T}$ may be easier than ${I \rightarrow I}$ because image-to-image retrieval is often the most difficult task. In the case of cross-modal task ${I \rightarrow T}$, the knowledge of text modality is transferred to image modality, making cross-modal retrieval easier. This shows cross-modal retrieval can be improved by knowledge transfer.

The precision-recall curves and the precision@top-$R$ curves \cite{cite:MM13LCMH,cite:VLDB14MSAE} are illustrated in Figure~\ref{fig:nuswide}. For space limitation, only the results of cross-modal tasks ${I \rightarrow T}$ and ${T \rightarrow I}$ are presented, while similar trends of results are observed on intra-modal tasks ${I \rightarrow I}$ and ${T \rightarrow T}$. CCQ shows the best cross-modal retrieval performance on all recall levels and top-$R$ ranks. This validates that CCQ is capable for diverse retrieval scenarios, which may emphasize higher precision at smaller number of top-$R$ retrieved results, i.e. Web search, or higher recall tolerating fairly lower precision, i.e. vertical search.

\begin{figure*}[!htb]
    \centering
    \subfigure[${I \rightarrow T}$ @ 16 bits]{
        \includegraphics[width=0.23\textwidth]{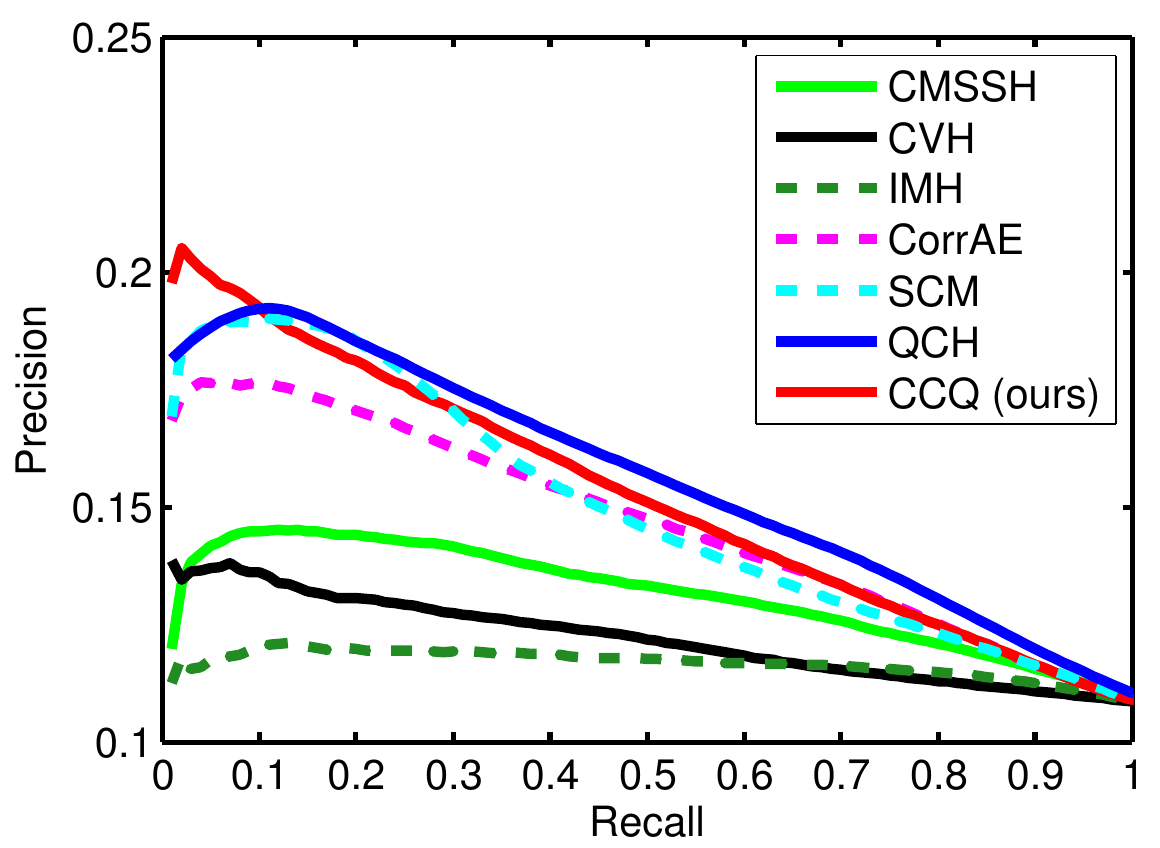}
        \label{fig:pr_wiki_1}
    }
    \subfigure[${I \rightarrow T}$ @ 32 bits]{
        \includegraphics[width=0.23\textwidth]{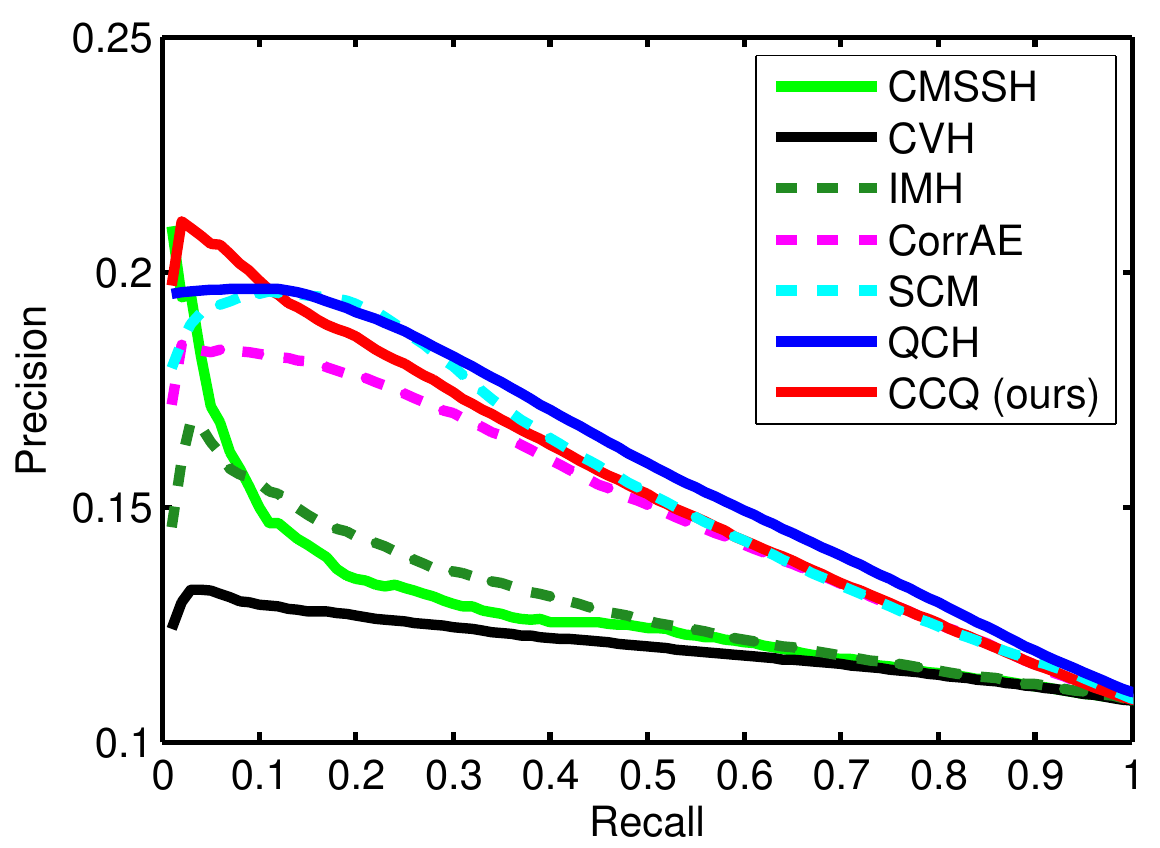}
        \label{fig:pr_wiki_2}
    }
    \subfigure[${T \rightarrow I}$ @ 16 bits]{
        \includegraphics[width=0.23\textwidth]{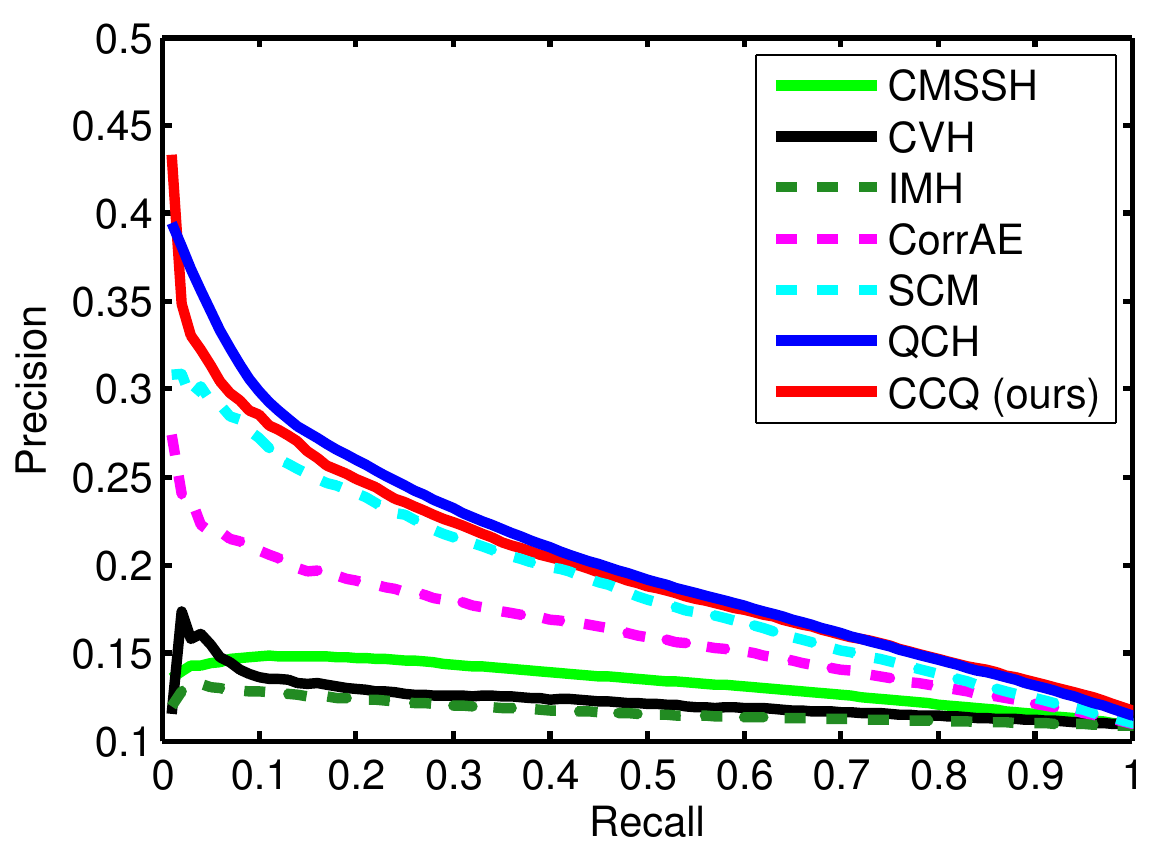}
        \label{fig:pr_wiki_3}
    }
    \subfigure[${T \rightarrow I}$ @ 32 bits]{
        \includegraphics[width=0.23\textwidth]{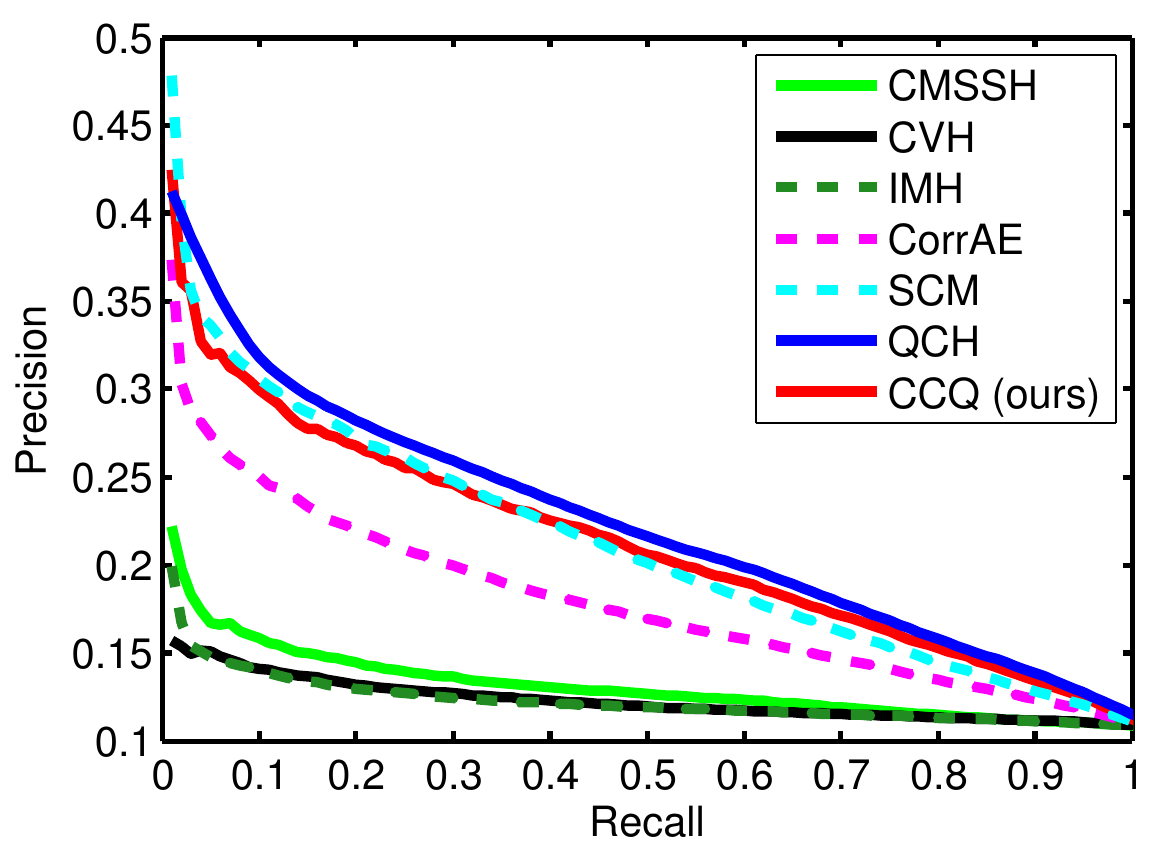}
        \label{fig:pr_wiki_4}
    }
    \\
    \vspace{-5pt}
    \subfigure[${I \rightarrow T}$ @ 16 bits]{
        \includegraphics[width=0.23\textwidth]{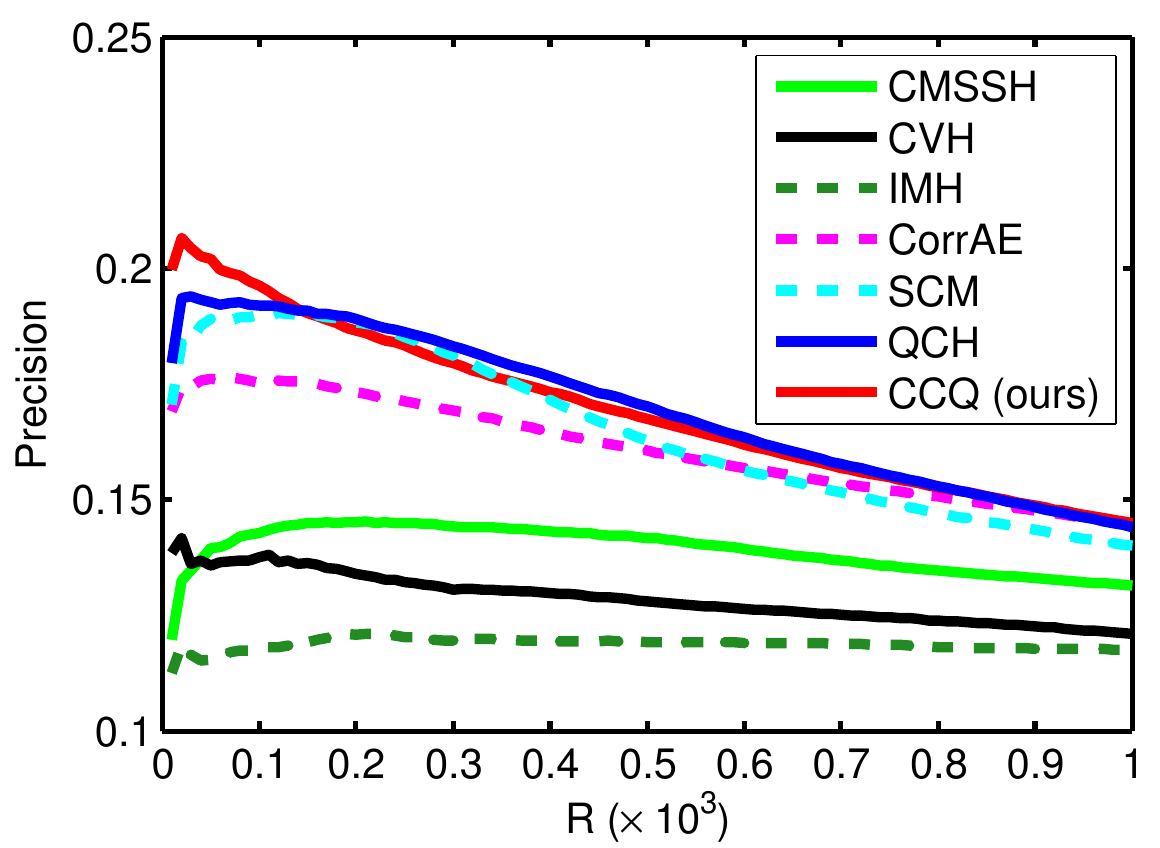}
        \label{fig:prec_wiki_1}
    }
    \subfigure[${I \rightarrow T}$ @ 32 bits]{
        \includegraphics[width=0.23\textwidth]{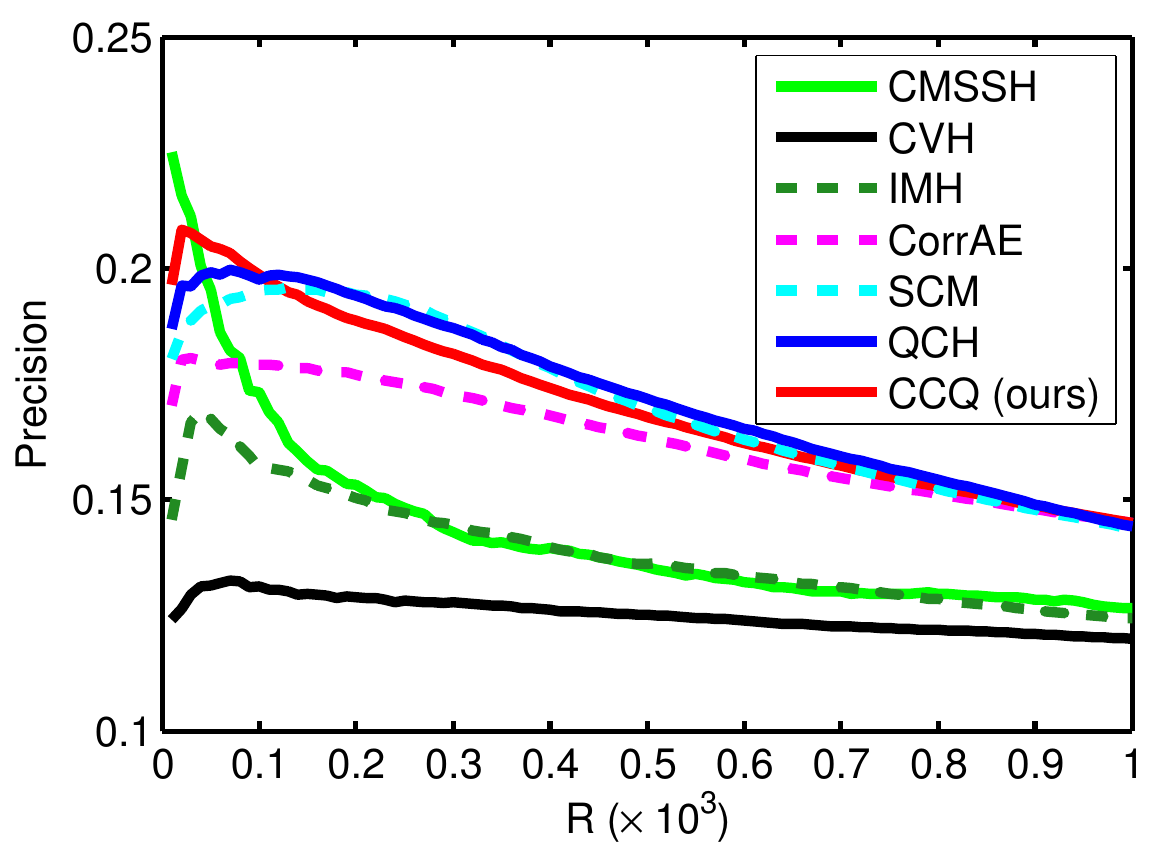}
        \label{fig:prec_wiki_2}
    }
    \subfigure[${T \rightarrow I}$ @ 16 bits]{
        \includegraphics[width=0.23\textwidth]{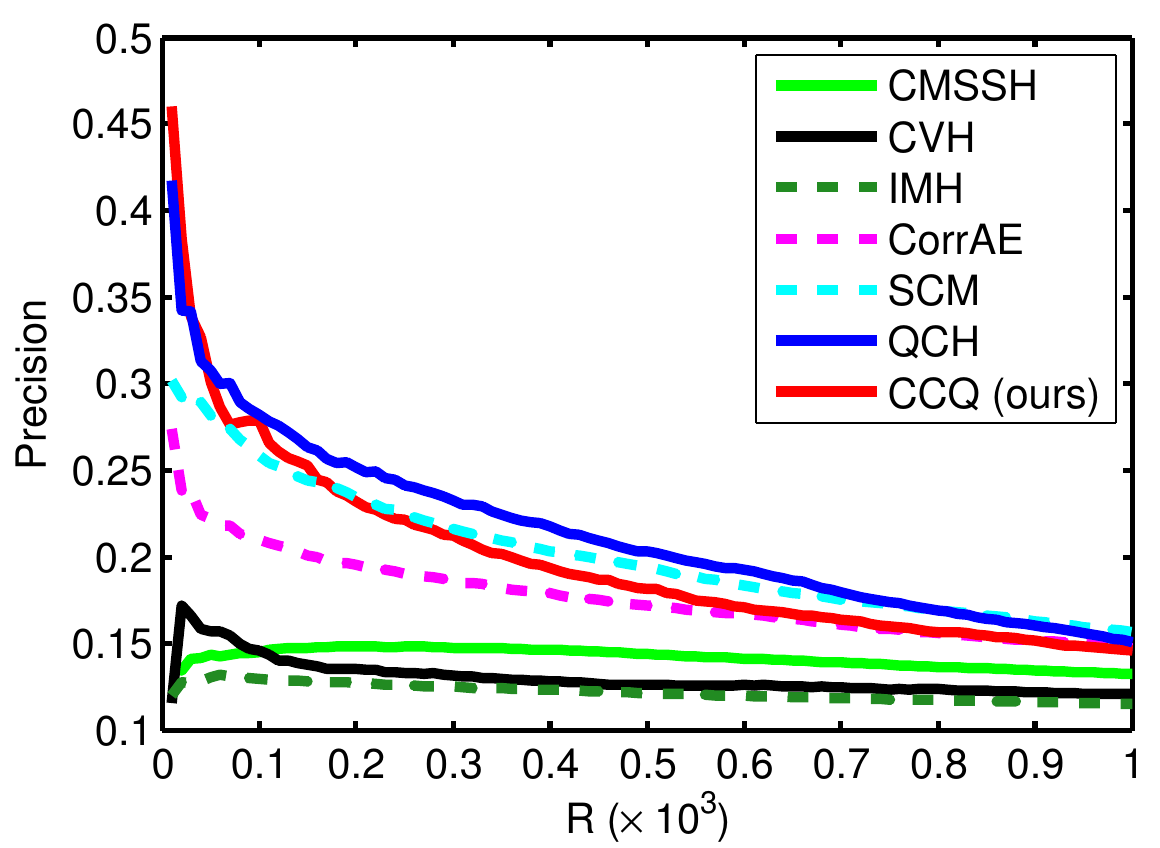}
        \label{fig:prec_wiki_3}
    }
    \subfigure[${T \rightarrow I}$ @ 32 bits]{
        \includegraphics[width=0.23\textwidth]{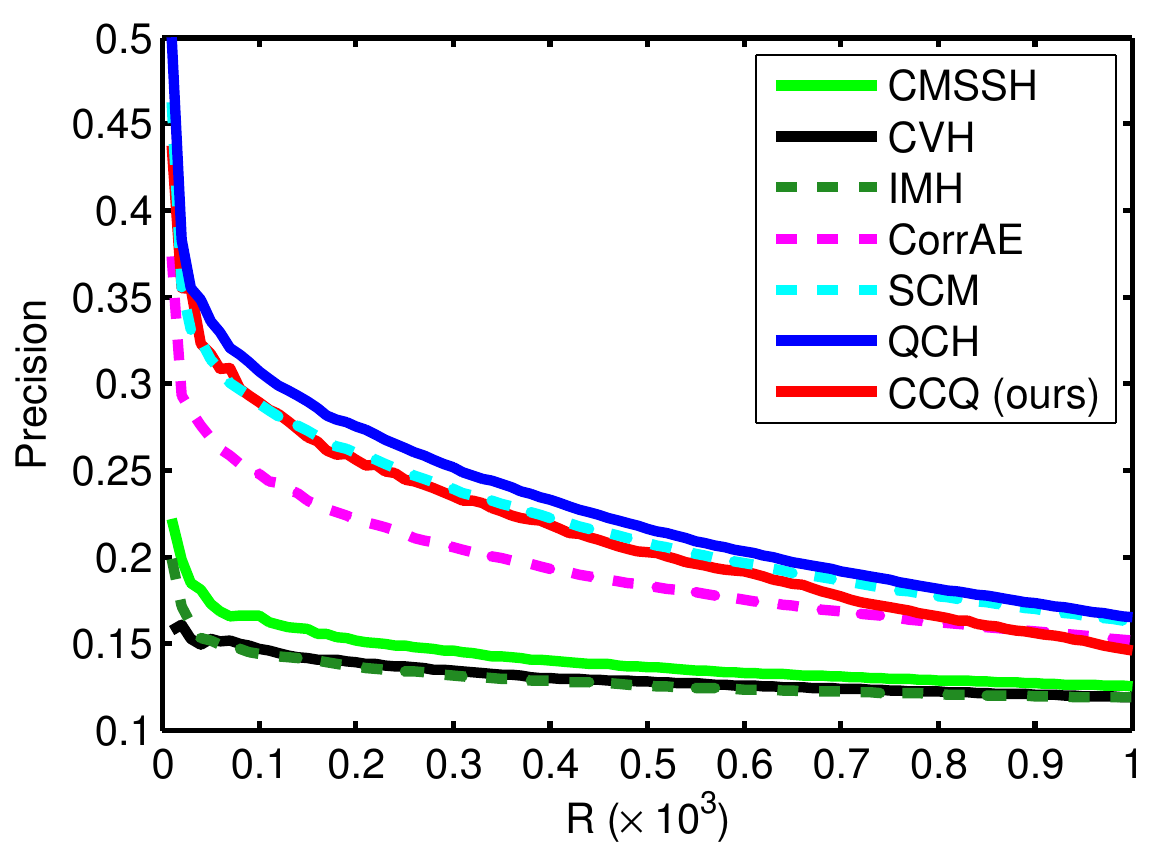}
        \label{fig:prec_wiki_4}
    }
    \vspace{-10pt}
    \caption{Precision-recall curves (top) and precision@R curves (bottom) on Wiki cross-modal search tasks @ 16 and 32 bits.}
    \label{fig:wiki}
\end{figure*}

\begin{figure*}[!htb]
    \centering
    \subfigure[${I \rightarrow T}$ @ 16 bits]{
        \includegraphics[width=0.23\textwidth]{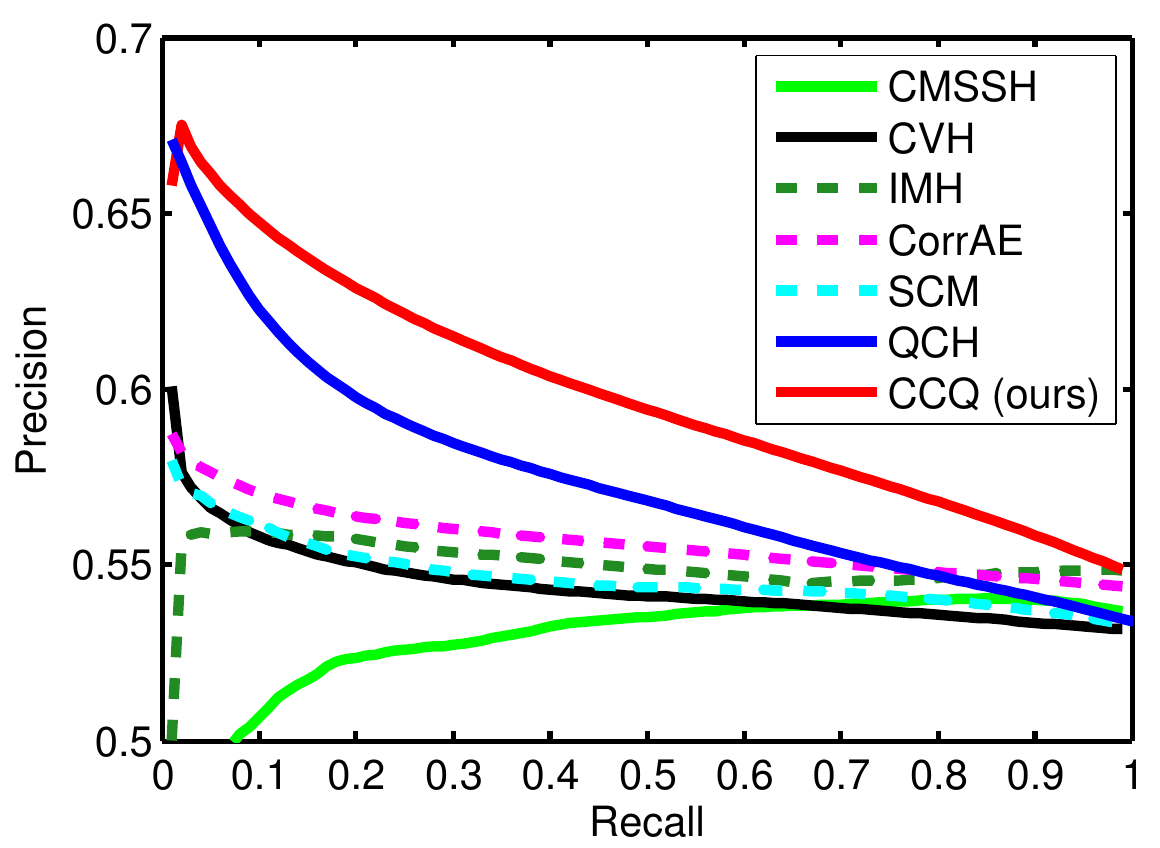}
        \label{fig:pr_flickr1m_1}
    }
    \subfigure[${I \rightarrow T}$ @ 32 bits]{
        \includegraphics[width=0.23\textwidth]{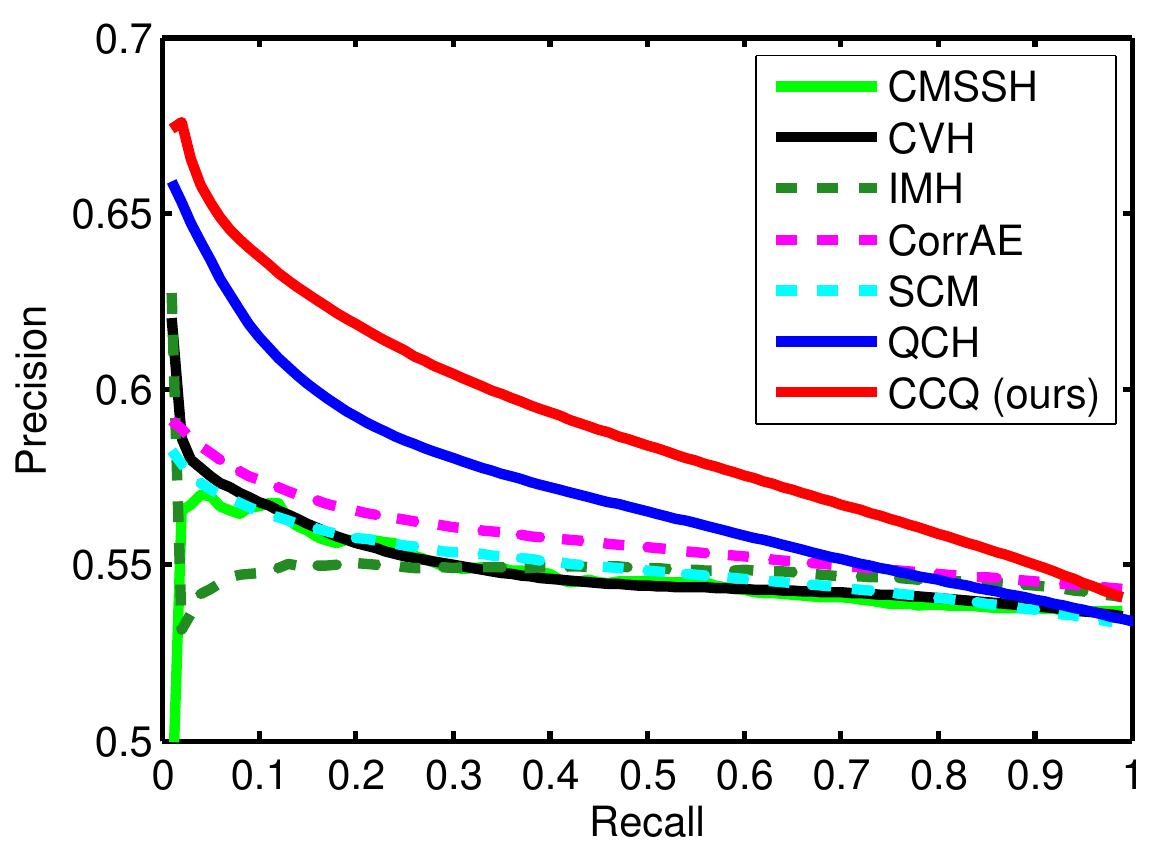}
        \label{fig:pr_flickr1m_2}
    }
    \subfigure[${T \rightarrow I}$ @ 16 bits]{
        \includegraphics[width=0.23\textwidth]{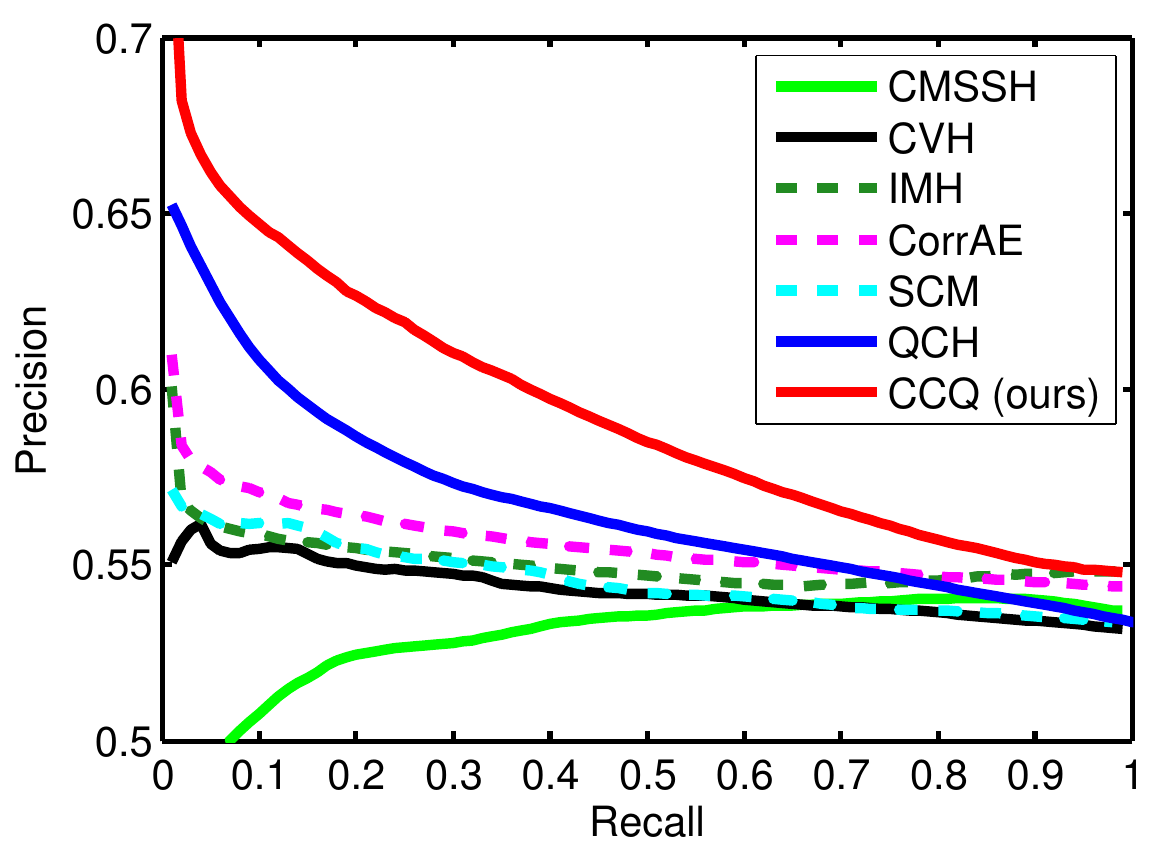}
        \label{fig:pr_flickr1m_3}
    }
    \subfigure[${T \rightarrow I}$ @ 32 bits]{
        \includegraphics[width=0.23\textwidth]{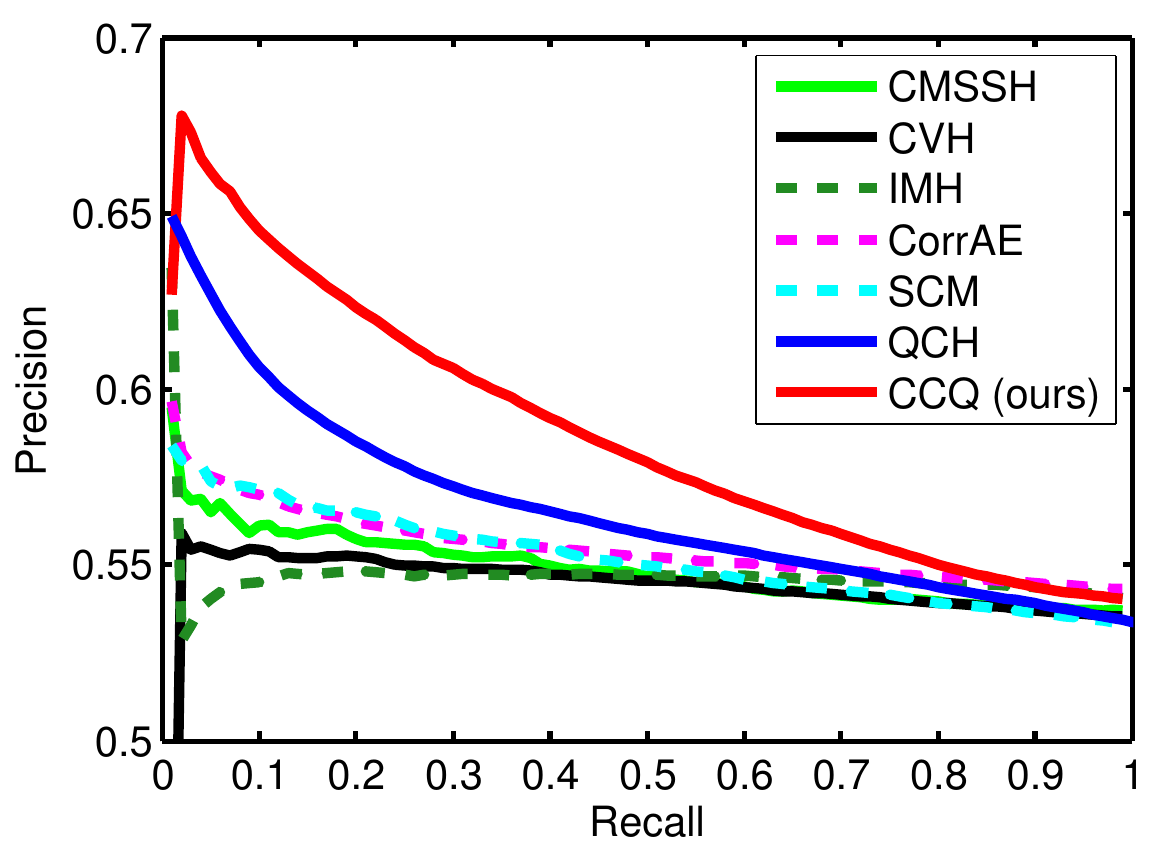}
        \label{fig:pr_flickr1m_4}
    }
    \\
    \vspace{-5pt}
    \subfigure[${I \rightarrow T}$ @ 16 bits]{
        \includegraphics[width=0.23\textwidth]{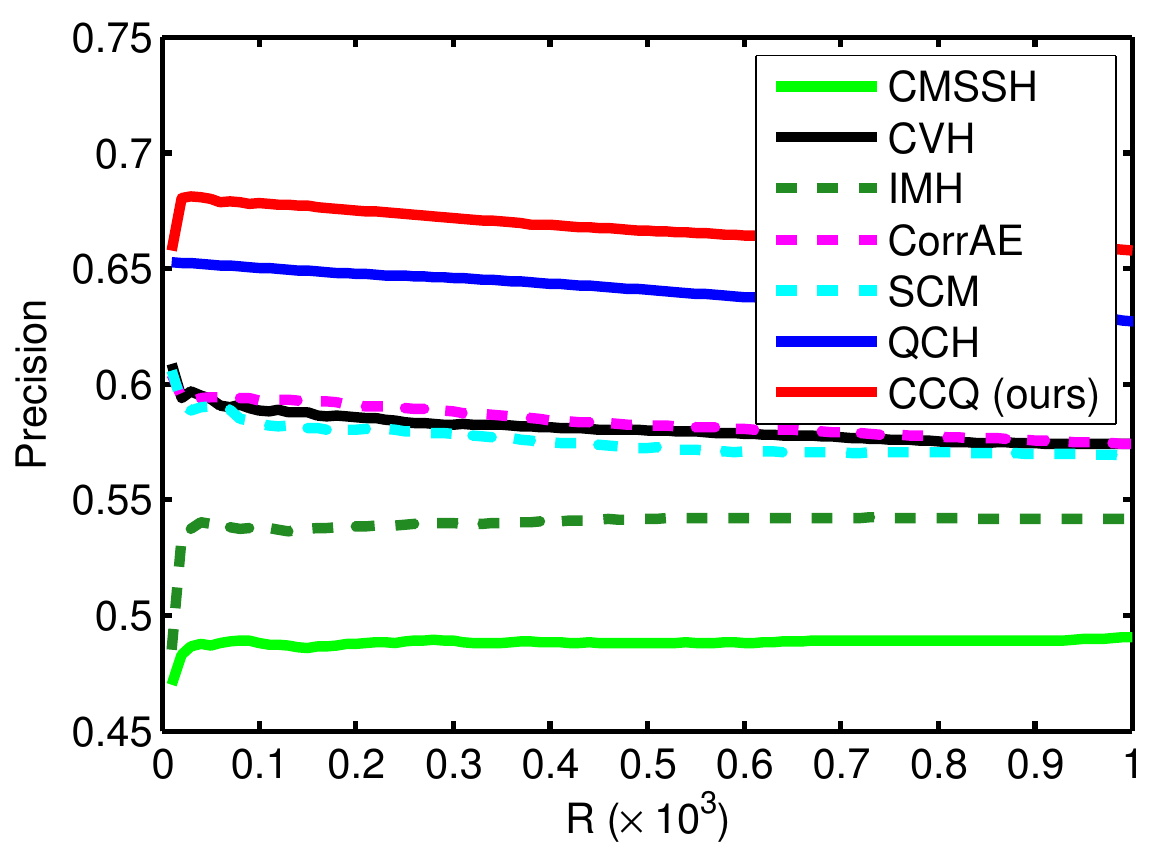}
        \label{fig:prec_flickr1m_1}
    }
    \subfigure[${I \rightarrow T}$ @ 32 bits]{
        \includegraphics[width=0.23\textwidth]{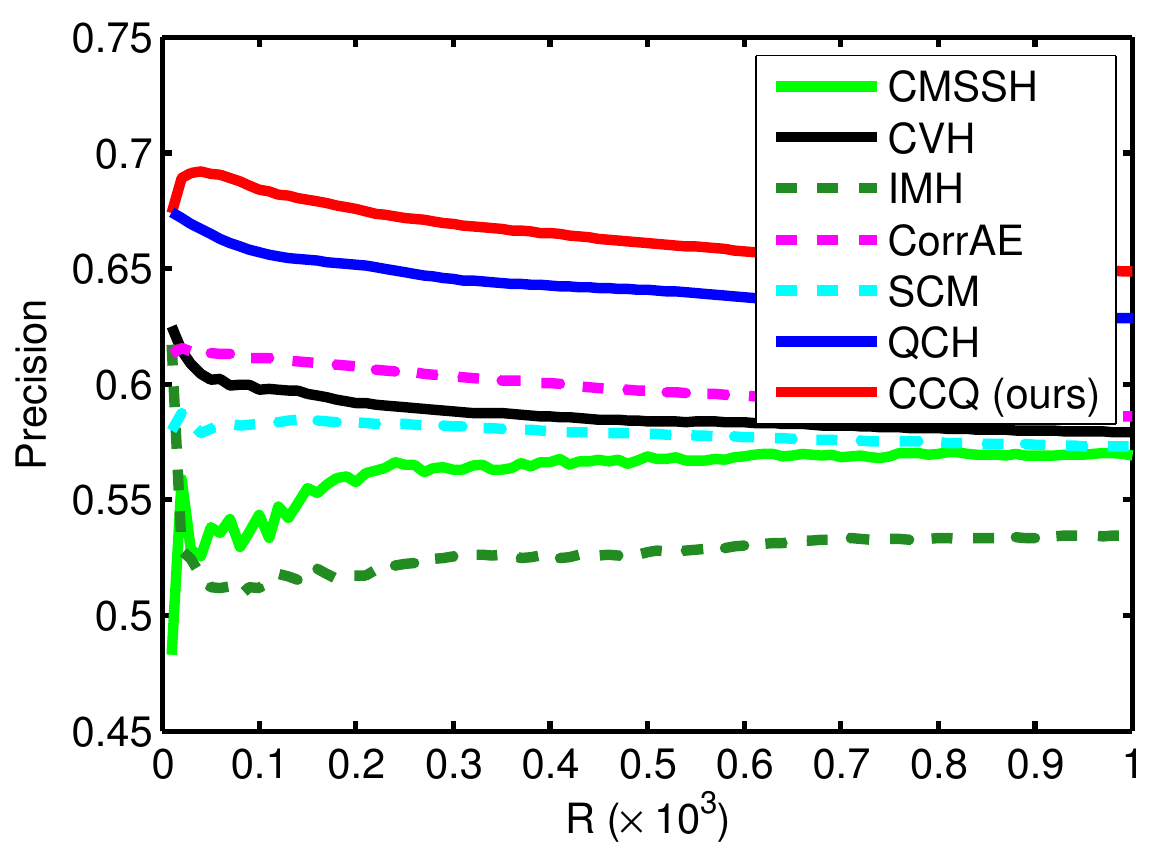}
        \label{fig:prec_flickr1m_2}
    }
    \subfigure[${T \rightarrow I}$ @ 16 bits]{
        \includegraphics[width=0.23\textwidth]{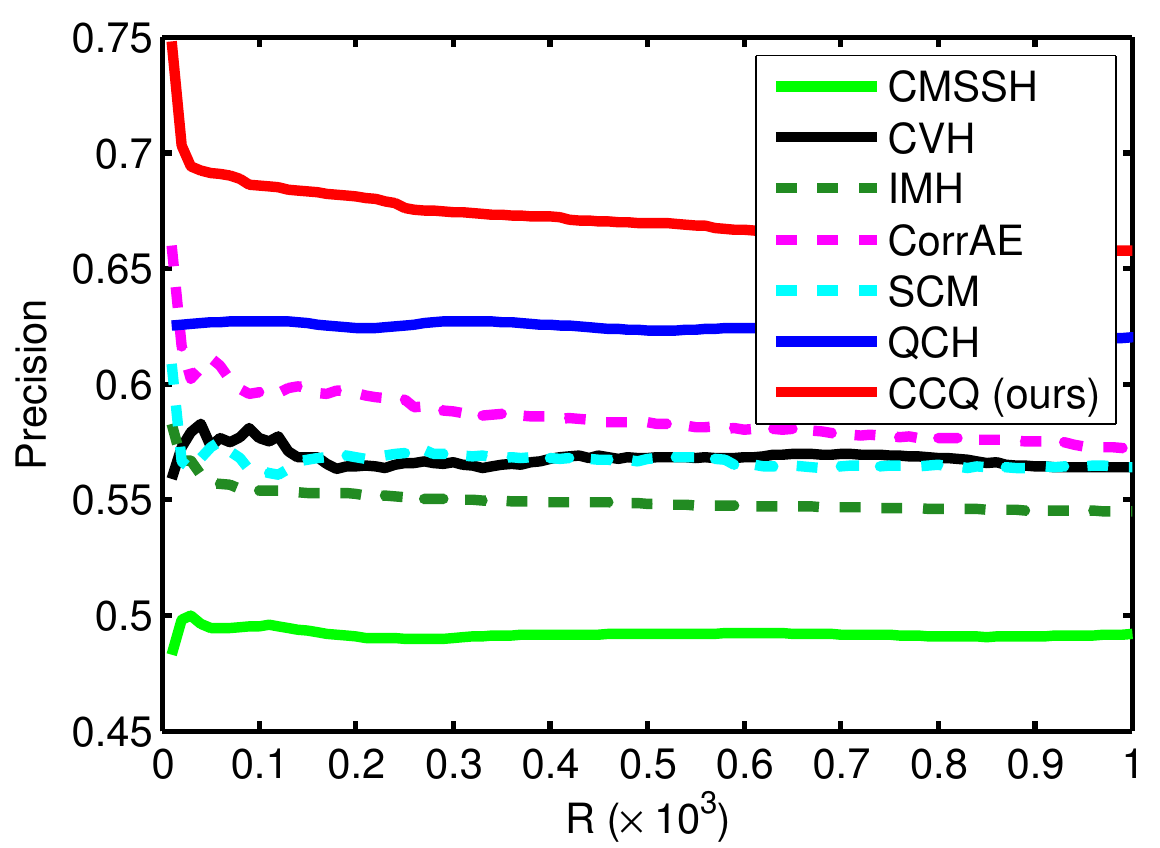}
        \label{fig:prec_flickr1m_3}
    }
    \subfigure[${T \rightarrow I}$ @ 32 bits]{
        \includegraphics[width=0.23\textwidth]{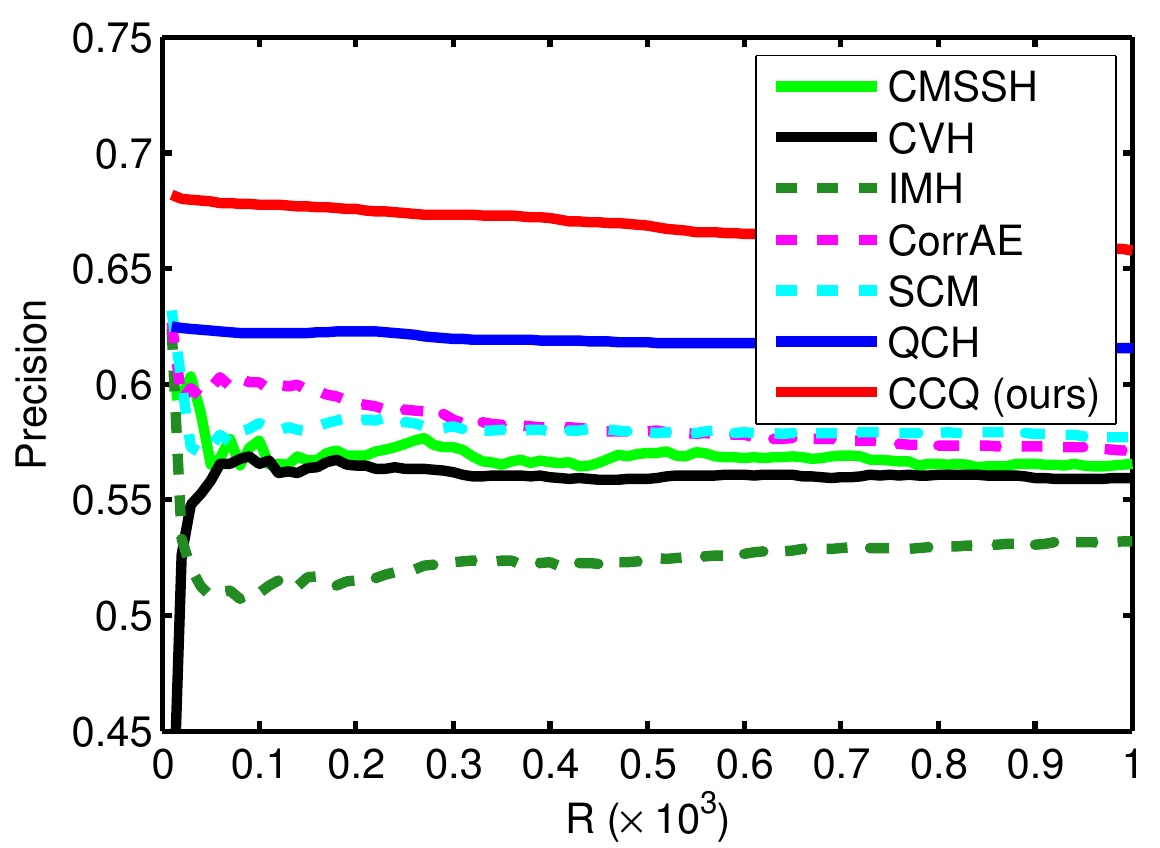}
        \label{fig:prec_flickr1m_4}
    }
    \vspace{-10pt}
    \caption{Precision-recall curves (top) and precision@R curves (bottom) on Flickr1M cross-modal search tasks @ 16 and 32 bits.}
    \label{fig:flickr1m}
\end{figure*}

\subsubsection{Results on Wiki}
Table~\ref{table:MAP} also compares the search performance of CCQ and the state of the art methods on the Wiki dataset, which shows that CCQ significantly outperforms the unsupervised hashing methods CVH, IMH, and CorrAE, and performs comparably to supervised hashing methods SCM and QCH. A notable observation is that the MAPs are much smaller than those on the NUS-WIDE dataset. This is reasonable as the images of Wiki are of low-quality (low-resolution) and high-diversity, i.e. the text can well describe the semantics of the image-text pair while the image may not be well related to the semantics of the image-text pair, which makes it more challenging to capture the semantic correlations between image query and text database. Note that the texts of Wiki are featured articles which are well edited by experts and rich in semantic information, hence it is fairly easy to correlate a text query with the multimodal database.

The precision-recall curves and the precision@top-$R$ curves \cite{cite:MM13LCMH,cite:VLDB14MSAE} are demonstrated in Figure~\ref{fig:wiki}. Again, CCQ is among the top-performing methods on all recall levels and all top-$R$ ranks. A noticeable performance drop can be examined from the precision-recall curves to the precision@top-$R$ curves. And this is because the Wiki dataset is very small-scale with only 2,173 database items, hence all relevant results will be retrieved at small $R$ and no more relevant results can be further retrieved when $R$ grows too large. This highlights the importance of evaluation with different metrics.

A crucial superiority of CCQ over the comparison methods lies in that CCQ jointly learns the isomorphic latent space and compact binary codes by minimizing both correlation and quantization errors in a unified optimization framework, while comparison methods merely learn the isomorphic space and binary codes in a separated two-step pipeline. As examined by CorrAE \cite{cite:MM14CorrAE}, the quality of searching with binary codes using Hamming distance is evidently inferior to searching with continuous features using Euclidean distance, due to substantial information loss by converting continuous features to binary codes without minimizing the quantization error. The search quality loss due to binarization is shown in Figure~\ref{fig:qerror}, and for CCQ, we use ${\bf R}^{v{\mathsf{T}}} {{\bf x}^v_n}$ for continuous features and ${\bf C}{\bf b}^v_n$ for binary codes. We see that IMH and CorrAE suffer from substantial MAP loss (similar trends are observed from other methods) while CCQ is almost lossless to binarization. In other words, by jointly minimizing the correlation error and quantization error, CCQ can circumvent information loss and learn more accurate binary codes.

\subsubsection{Results on Flickr1M}
In practical retrieval systems, it is crucial to process large-scale datasets in both training and testing phases, and thus we compare CCQ with state of the art methods on large-scale Flickr1M dataset. We report the MAP results in Table~\ref{table:MAP} and illustrate the detailed precision-recall curves and precision@top-$R$ curves in Figure~\ref{fig:flickr1m}. As mentioned before, we randomly select 10,000 image-text pairs as training set to learn hash functions if it is computationally too demanding to train these methods on the complete Flickr1M dataset. We can observe that CCQ significantly outperforms the comparison methods on all retrieval tasks and performs better with longer codes. This validates the superiority of CCQ in processing large-scale datasets, as the experimental setting on Flickr1M is consistent with real-word system setting where a sufficiently accurate model needs to be derived on a sufficiently large training set. We will examine CCQ's ability to process real semi-paired data in the sequel.

\begin{figure*}[!htb]
    \centering
    \subfigure[MAP Loss]{
        \includegraphics[width=0.23\textwidth]{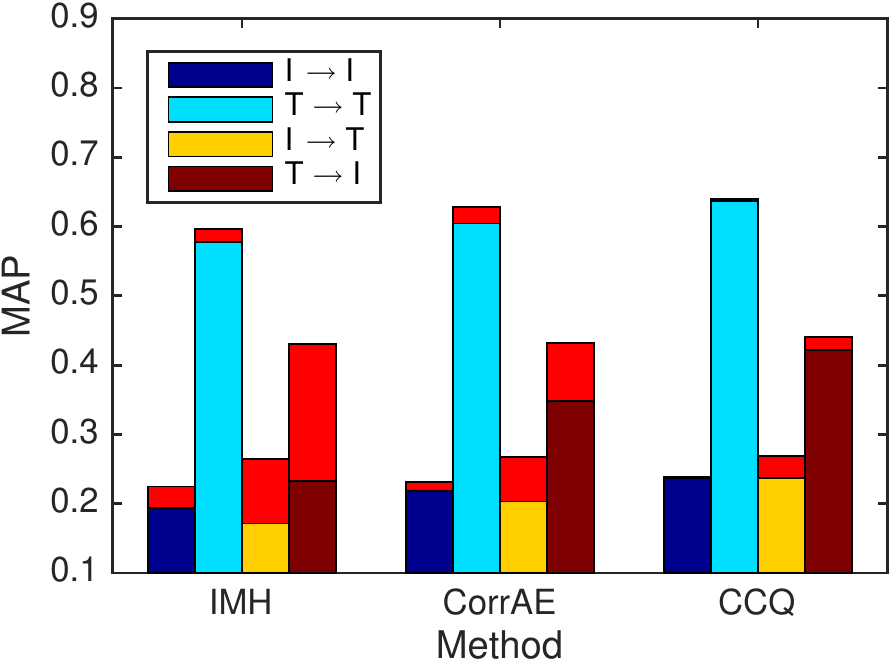}
        \label{fig:qerror}
    }
    \subfigure[NUS-WIDE]{
        \includegraphics[width=0.23\textwidth]{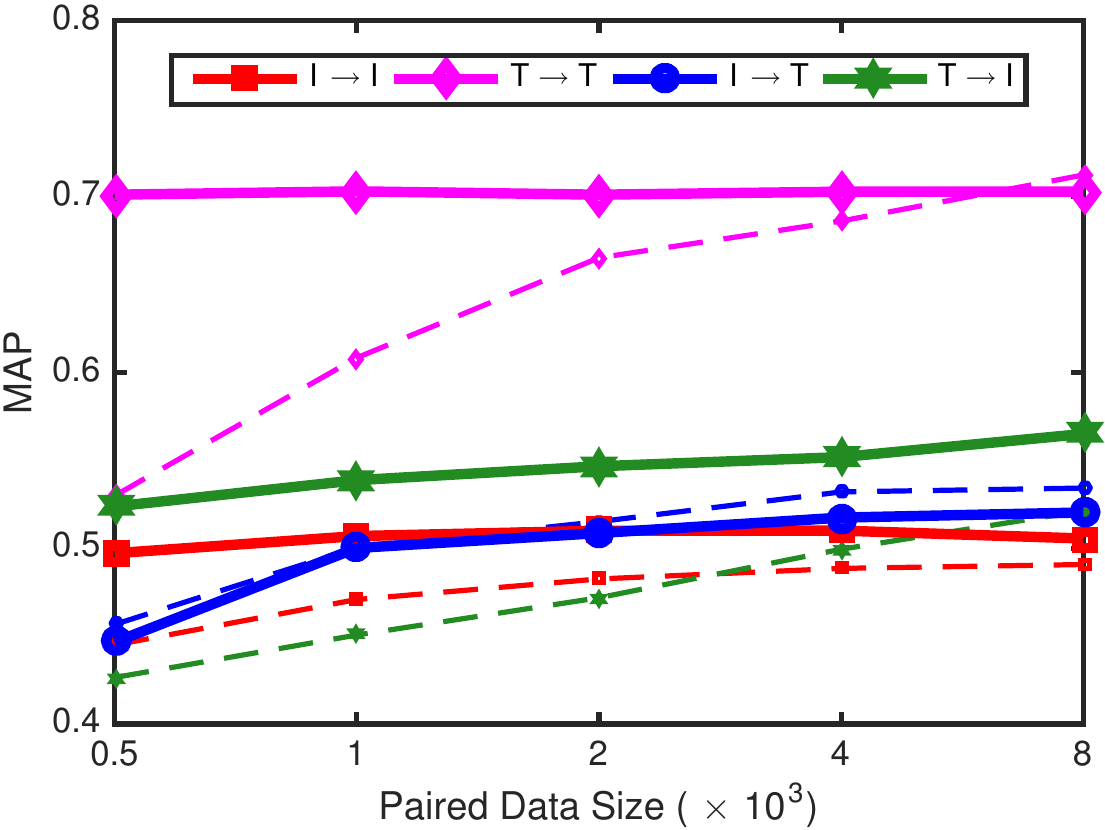}
        \label{fig:semi_pair_nuswide}
    } 
    \subfigure[Flickr1M]{
        \includegraphics[width=0.23\textwidth]{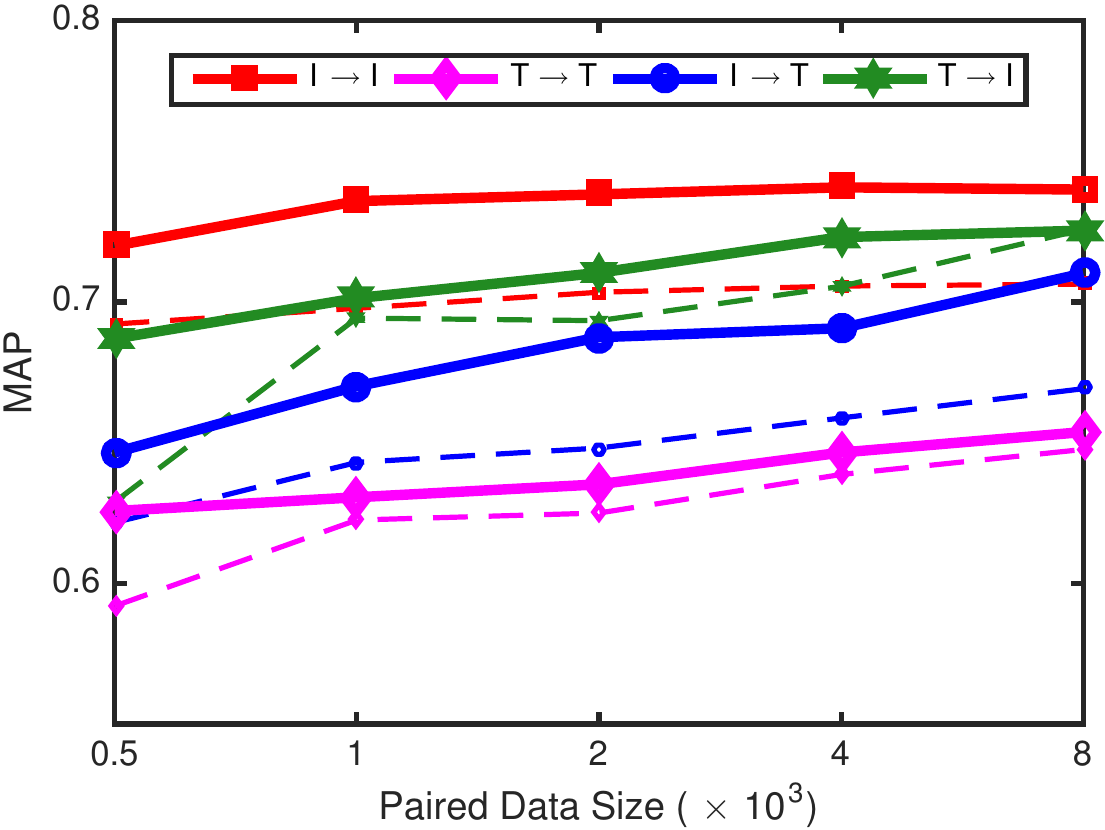}
        \label{fig:semi_pair_flickr1m}
    }
    \subfigure[Search Efficiency]{
        \includegraphics[width=0.23\textwidth]{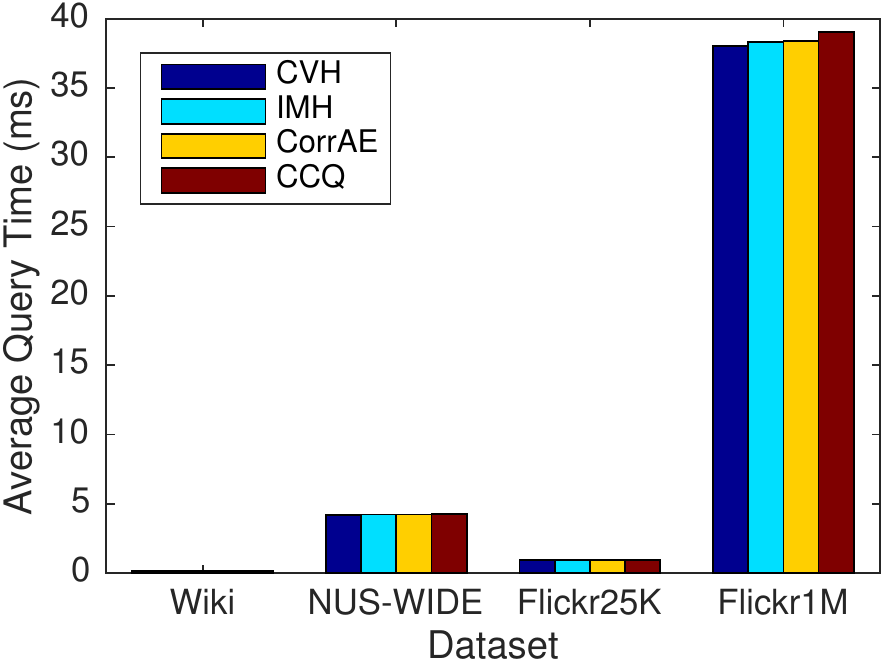}
        \label{fig:search}
    }
    \vspace{-10pt}
    \caption{Effectiveness and efficiency experiments: (1) Loss of search quality in MAP (by red bars) due to conversion from continuous features to binary codes on Wiki. (b)--(c) the MAP of CCQ w.r.t. different numbers of paired data points (the number of unpaired data points is fixed to $10,000$). Solid lines indicate training with both paired and unpaired data, and dashed lines indicate training with only paired data. (d) Average search time (ms) for each query via lookup tables on Wiki, NUS-WIDE, Flickr25K, and Flickr1M.}
\end{figure*}

\begin{figure*}[!htb]
    \centering
    \subfigure[Time]{
        \includegraphics[width=0.23\textwidth]{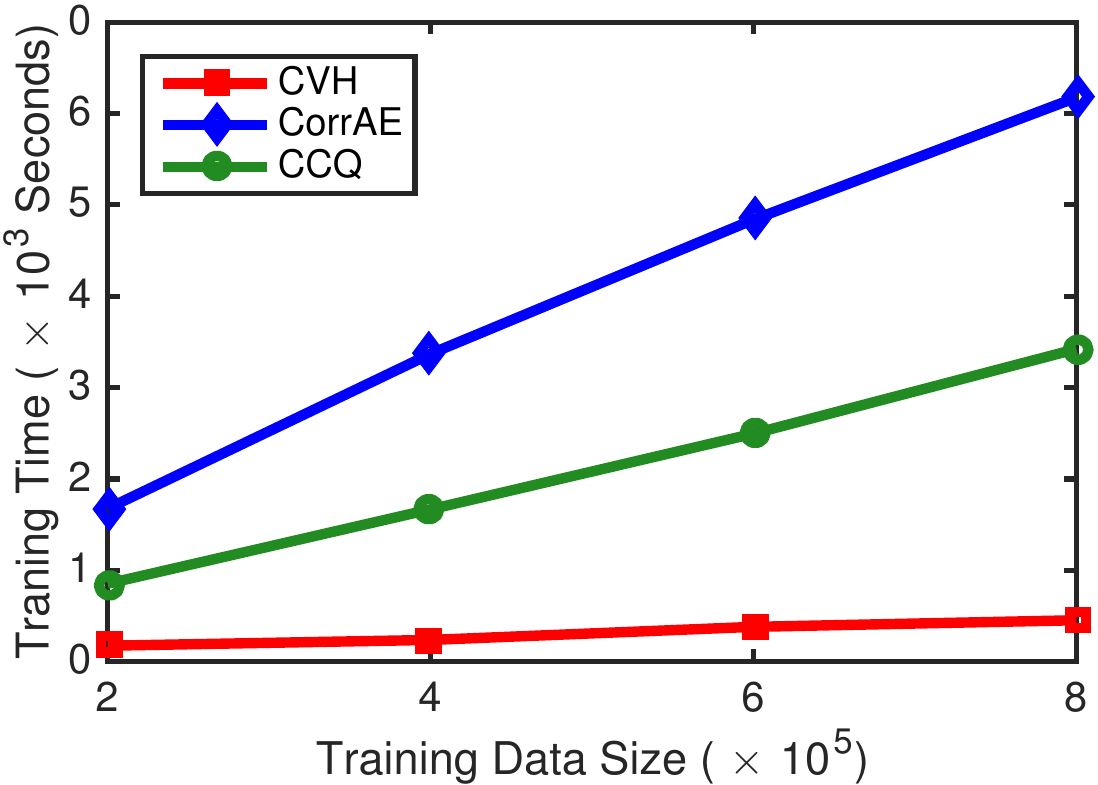}
        \label{fig:time}
    }
    \subfigure[Memory]{
        \includegraphics[width=0.23\textwidth]{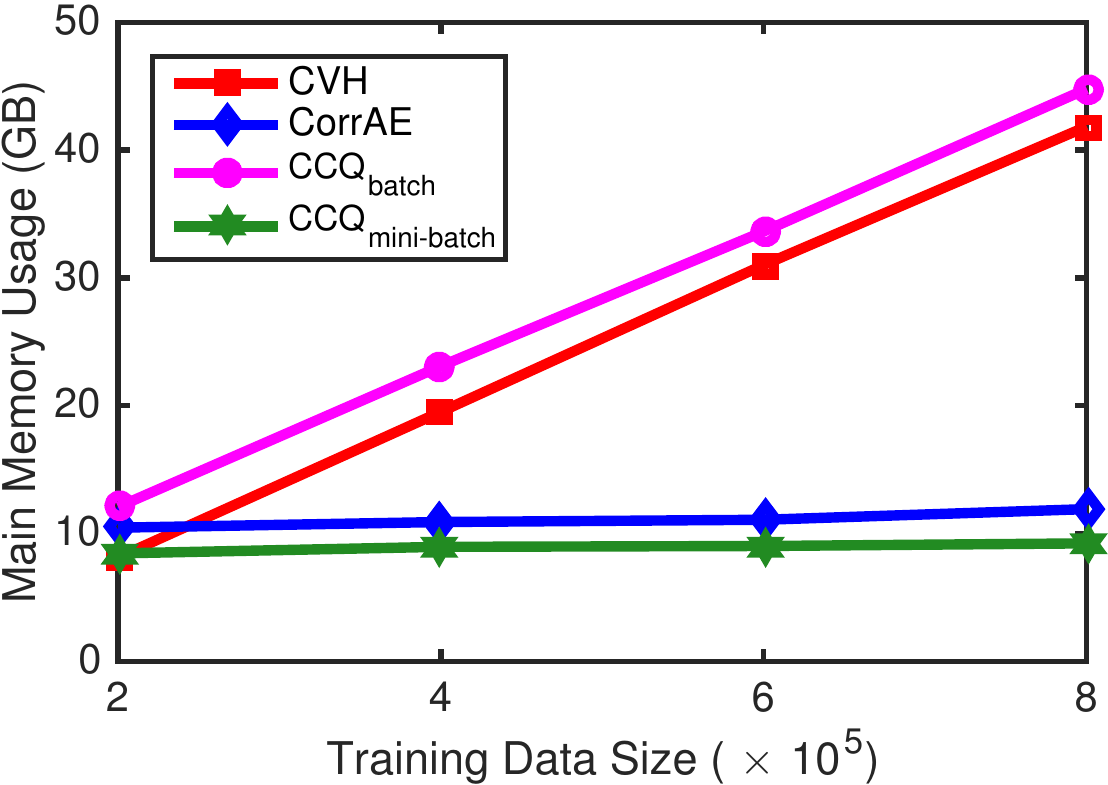}
        \label{fig:memory}
    }
    \subfigure[${I \rightarrow T}$ @ 32 bits]{
        \includegraphics[width=0.23\textwidth]{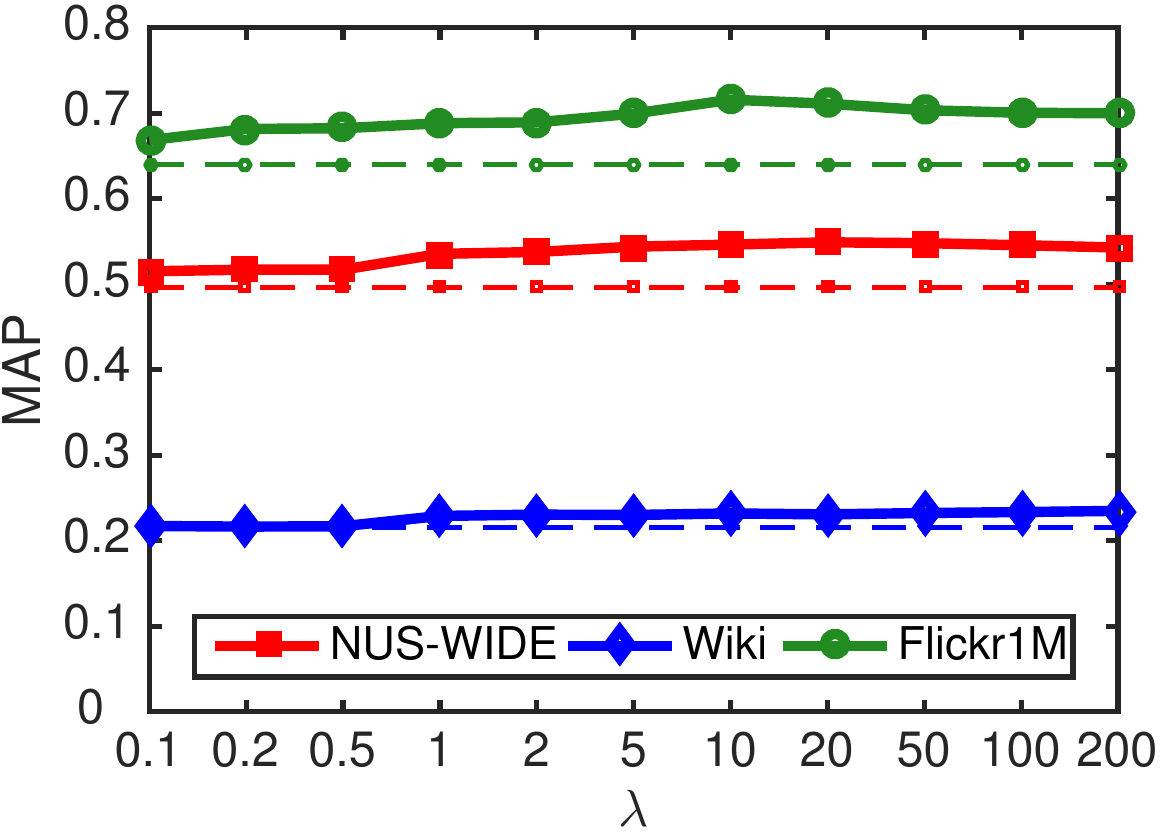}
        \label{fig:lambda_it}
    }
    \subfigure[${T \rightarrow I}$ @ 32 bits]{
        \includegraphics[width=0.23\textwidth]{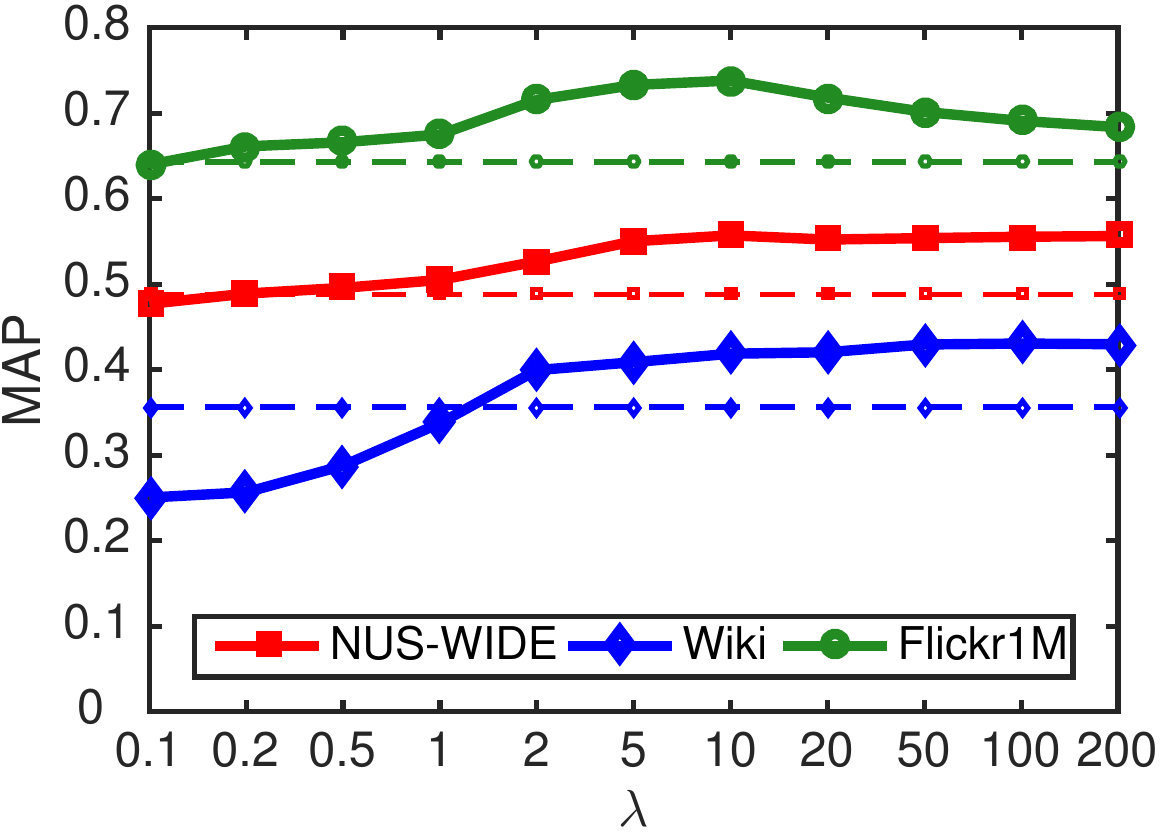}
        \label{fig:lambda_ti}
    }
    \vspace{-10pt}
    \caption{Efficiency verification experiments: (a)--(b) Training time and memory costs of different methods on the complete Flickr1M dataset. CCQ with batch (mini-batch) training scales linearly (constantly) to the sample size. (c)--(d) The MAP of CCQ @ $32$ bits versus parameter $\lambda \in [0.1,200]$ for cross-modal retrieval tasks $I \rightarrow T$ and $T \rightarrow I$ on the NUS-WIDE, Wiki, and Flickr1M datasets.}
    \label{fig:sensitivity}
\end{figure*}

\subsection{Semi-Paired Data Quantization}
Most of the existing methods, including the ones in comparison, require full correspondences between different modalities. In other words, the multimodal data objects are fully paired, e.g. image-text pairs. As a result, these methods are unable to tackle more realistic scenarios in which only a limited number of paired data points are available. CCQ explores the idea of  semi-supervised learning and can leverage both paired data (all modalities of the objects are available) and unpaired data (partial modalities of the objects are available) to boost the search quality when paired data are limited. To verify this, we consider the NUS-WIDE and Flickr1M datasets and for each dataset, we randomly sample as the training set 1) 10,000 images without text modality, 2) 10,000 texts without image modality, and 3) different numbers, i.e. $[0.5, 1, 2, 4, 8] \times 10^3$, of image-text pairs. We train CCQ with these semi-paired data and evaluate the search performance in terms of MAP @ 32 bits.

The search performances of CCQ on NUS-WIDE and Flickr1M are demonstrated in Figures~\ref{fig:semi_pair_nuswide} and \ref{fig:semi_pair_flickr1m} respectively, where solid lines indicate training with both paired and unpaired data, and dashed lines indicate training with only paired data. We can observe that when the number of paired data points is small, CCQ trained with both paired and unpaired data significantly outperforms CCQ trained with only paired data on most of the multimodal search tasks; when the number of paired data points increases, the search performance of CCQ will gradually saturate while the search quality of the two training paradigms will finally match. This clearly shows that CCQ can effectively leverage both paired and unpaired data (partial multimodal data) to boost search quality in a semi-paired data scenario. 

An unexpected phenomenon is that semi-paired training slightly deteriorates search performance on task ${I \rightarrow T}$. We conjecture the plausible reason is that searching text database with image queries significantly relies on maximizing the image-text correlations to bridge the semantic gap between low-level image features and high-level image semantics, i.e. its associated texts. When the number of paired data points is obviously smaller than the number of unpaired data points, semi-paired training may tend to weaken correlation learning from image-text pairs and incur performance degradation.

\subsection{Search Efficiency}
To search for approximate nearest neighbors (ANN) in database for a given query, all methods in comparison perform linear scan using symmetric or asymmetric distance. Specifically, to compare a query vector with a database vector, CVH, IMH, and CorrAE all compute symmetric Hamming distance via lookup tables, and CCQ constructs a distance lookup table for each query that stores the Euclidean distances between the query and the multiple codebooks. As a result, CVH, IMH, CorrAE, and CCQ compute exactly the same number of table lookups for linear scan, while their costs of computing the query-codebook distance lookup tables are slightly different, which can be negligible as they are infinitesimal w.r.t. the cost of linear scan. For example, the cost of computing the distance lookup table for CCQ takes only less than 1\% of the cost for linear scan on Flickr1M. The average search time of each query by CVH, IMH, CorrAE, and CCQ on the Wiki, NUS-WIDE, Flickr25K, and Flickr1M datasets is illustrated in Figure~\ref{fig:search}, from which we can observe that the search efficiency are comparable for all methods. While it is beyond the scope of this paper, we want to note that one can adopt a Multi-Index \cite{cite:CVPR12IMI} approach to achieve sub-linear search complexity on the binary codes and further boost search efficiency.

\subsection{Training Complexity}\label{section:scalability}
The training time and memory costs of CCQ scale linearly with the training sample size and hence can process large-scale dataset. To verify this, we follow \cite{cite:VLDB14MSAE} and use the complete Flickr1M dataset to evaluate the consumptions of training time and memory. CMSSH and IMH are not compared in this study since they require $O(N^2)$ complexity and run out of either time or memory on this dataset.

The comparison of training time costs is illustrated in Figure~\ref{fig:time}. We can observe that the training time of CCQ increases linearly with respect to the sample size. Due to multiple iterations between three sets of variables, i.e. transformation matrices ${\bf R}^v$, quantizer codebook $\bf{C}$, and modal-specific binary codes ${\bf B}^v$, CCQ is not as efficient as CVH. However, CCQ performs much more efficiently in time than CorrAE, which is a deep learning based method solving a time-demanding non-convex nonlinear optimization problem.

The training memory consumptions are compared in Figure~\ref{fig:memory}. Both batch and mini-batch (large-scale) implementations of CCQ store the model parameters in memory, which are independent of training dataset size. For the batch implementation, all training data is loaded in memory, while for the mini-batch implementation, the training data is partitioned into multiple mini-batches while only one mini-batch is loaded in memory each time. Hence in the mini-batch (large-scale) implementation, the memory cost stays constant when training dataset size increases. We can flexibly allocate memory to each mini-batch to trade off memory and disk reading costs.

\subsection{Parameter Sensitivity}\label{section:Sensitivity}
Towards unsupervised multimodal retrieval, CCQ is designed to involve only two  parameters, dimension of modality-isomorphic subspace $D$ and modality trade-off weight $\lambda$, and the performance is expected to be stable against parameter variations. Since we have fixed $D=\min(\{P_v\}_{v=1}^V,H)$, we only inspect the sensitivity of $\lambda$.

We compute MAP @ 32 bits on both cross-modal retrieval tasks by varying $\lambda$ between 0.1 and 200. The performance of CCQ w.r.t. parameter $\lambda$ is shown in Figure~\ref{fig:lambda_it} and \ref{fig:lambda_ti}. We see that CCQ can consistently outperform all the unsupervised baseline methods by a large margin with $\lambda$ varying between $1$ and $200$. This validates that CCQ is robust against parameter selection and is applicable to unsupervised multimodal retrieval with easily-configured parameters.

\section{Conclusion and Future Work}\label{section:Conclusion}
In this paper, we have formally approached seamless multimodal hashing through a novel composite correlation quantization (CCQ). It integrates multimodal correlation and composite quantization into a seamless latent semantic analysis (LSA) framework, which yields compact binary codes that encode both intra-modal similarity and inter-modal correlation. The sharing of codebooks and binary codes across modalities enables joint learning of latent semantics that are maximally correlated in the isomorphic feature space, which serves as the key contributor to the efficacy of the proposed CCQ method.

In the future, we plan to equip our model with a deep learning architecture which can learn highly abstract nonlinear representations to better distill the correlation structures across multiple modalities.

\section{Acknowledgments}
The authors would like to thank Dr Jingdong Wang for insightful comments. This work was supported by National Natural Science Foundation of China (61325008, 61502265), China Postdoctoral Science Foundation (2015T80088), National Science\&Technology Supporting Program (2015BAH14F02), NSF grant III-1526499, and Tsinghua TNList Lab Fund for Big Data Science and Technology.

\bibliographystyle{abbrv}
\bibliography{CCQ}

\end{document}